\documentclass[iicol,sn-mathphys-num]{sn-jnl}% Math and Physical Sciences Numbered Reference Style 

\usepackage{graphicx}%
\usepackage{multirow}%
\usepackage{amsmath,amssymb,amsfonts}%
\usepackage{amsthm}%
\usepackage{mathrsfs}%
\usepackage[title]{appendix}%
\usepackage{xcolor}%
\usepackage{textcomp}%
\usepackage{manyfoot}%
\usepackage{booktabs}%
\usepackage{algorithm}%
\usepackage{algorithmicx}%
\usepackage{algpseudocode}%
\usepackage{listings}%
\usepackage{bm}
\usepackage{subcaption}  % 用于子图
\usepackage{float}  % 加载 float 包以使用 [H]
\usepackage{placeins}
\usepackage{algorithm}
\usepackage{algpseudocode}
\newtheorem{Thm}{\textbf{Theorem}}

\newtheorem{Def}{\textbf{Definition}}
\newtheorem{Rmk}{\textbf{Remark}}
\newtheorem{Prob}{\textbf{Problem}}

\newtheorem{Asum}{\textbf{Assumption}}

\renewcommand{\bf}{\mathbf} % 重定义 \bf 为 \mathbf
\renewcommand{\rm}{\mathrm} % 重定义 \bf 为 \mathbf
\theoremstyle{thmstyleone}
\theoremstyle{thmstyletwo}%
\theoremstyle{thmstylethree}%


\begin{document}

\title[Article Title]{Multi-Robot Cooperative Herding through Backstepping Control Barrier Functions}
% {Safe Repulsion-Driven Swarm Herding through Backstepping Control Barrier Functions}
\author[1,2]{Kang Li}\email{likang32@westlake.edu.cn} % \textsuperscript{\dag}
\equalcont{These authors contributed equally to this work.}
\author[3]{Ming Li }\email{ming3@kth.se} % \textsuperscript{\dag}
\equalcont{These authors contributed equally to this work.}
\author[2]{Wenkang Ji}\email{jiwenkang@westlake.edu.cn}
% \author*[4]{Zhiyong Sun}\email{zhiyong.sun@pku.edu.cn}
% \author*[2]{Shiyu Zhao}\email{zhaoshiyu@westlake.edu.cn}
\author*[4]{Zhiyong Sun}\email{zhaoshiyu@westlake.edu.cn}
\author*[2]{Shiyu Zhao}\email{zhiyong.sun@pku.edu.cn}

\affil[1]{
  \orgdiv{College of Computer Science and Technology}, 
  \orgname{Zhejiang University}, 
  \orgaddress{
    \city{Hangzhou}, 
    \postcode{310058}, 
    \state{Zhejiang},                
    \country{China}}}
\affil*[2]{
  \orgdiv{Windy Lab}, 
  \orgname{Department of Artificial Intelligence}, 
  \orgname{Westlake University}, 
  \orgaddress{
    \city{Hangzhou}, 
    \postcode{310030}, 
    \state{Zhejiang},             
    \country{China}}}
\affil[3]{
  \orgdiv{Division of Decision and Control Systems}, 
  \orgname{KTH Royal Institute of Technology}, 
  \orgaddress{
    \city{Stockholm}, 
    \postcode{100 44}, 
    \country{Sweden}}}
\affil[4]{
  \orgdiv{Department of Mechanics and Engineering Science}, 
  \orgname{College of Engineering}, 
  \orgname{Peking University}, 
  \orgaddress{
    \city{Beijing}, 
    \postcode{100871}, 
    \state{Beijing}, 
    \country{China}}}

%-----------------------------------------------------------------------------------------------------
\abstract
{
We propose a novel cooperative herding strategy through backstepping control barrier functions (CBFs), which coordinates multiple herders to herd a group of evaders safely towards a designated goal region.
For the herding system with heterogeneous groups involving herders and evaders, the behavior of the evaders can only be influenced indirectly by the herders' motion, especially when the evaders follow an inverse dynamics model and respond solely to repulsive interactions from the herders. 
This indirect interaction mechanism inherently renders the overall system underactuated. 
To address this issue, we first construct separate CBFs for the dual objectives of goal reaching and collision avoidance, which ensure both herding completion and safety guarantees. 
Then, we reformulate the underactuated herding dynamics into a control-affine structure and employ a backstepping approach to recursively design control inputs for the hierarchical barrier functions, avoiding taking derivatives of the higher-order system.
Finally, we present a cooperative herding strategy based on backstepping CBFs that allow herders to safely herd multiple evaders into the goal region. 
In addition, centralized and decentralized implementations of the proposed algorithm are developed, further enhancing its flexibility and applicability. 
Extensive simulations and real-world experiments validate the effectiveness and safety of the proposed strategy in multi-robot herding.
}


\maketitle

% \newpage
% \clearpage
% ----------------------------------------------------------------------------------------------
\section{Introduction}
Herding is a common form of collective behavior observed in nature~\cite{king2012selfish, ringhofer2020herding, coppinger2000dogs}, involving one or more herders (e.g., shepherds) directing the motion of evaders (e.g., sheep) through interactions often modeled as repulsive forces~\cite{kubo2022herd, singh2024planar}. 
This process typically aims to herd the evader group from an initial region toward a target destination.
Such behavior exemplifies complex interactions involving both cooperation and antagonism, highlighting challenges in distributed coordination and decision-making. 
These characteristics have inspired extensive research in multi-robot systems, leading to a variety of collective tasks, such as formation~\cite{huang2022distributed, bai2021learning, zhao2018affine, zhao2019bearing}, clustering~\cite{martin2023multi, sun2023mean, li2020distributed}, encirclement~\cite{zhang2022multi}, and herding~\cite{paranjape2018robotic, zhang2024distributed, sebastian2022adaptive}. 
In contrast to these well-explored tasks, this study focuses on the problem of safe swarm herding with multiple herders, which poses significant difficulties due to both its theoretical complexity and practical constraints. 

Herding primarily involves non-cooperative interactions between two heterogeneous groups of herders and evaders.
These groups differ significantly in dynamics, control strategies, and behavioral objectives, which presents several challenges for multi-robot cooperative control. 
First, the heterogeneity in dynamics and coordination~\cite{sebastian2022adaptive} arises from the fact that evaders typically exhibit passive and stimulus-driven behavior while herders are actively controlled, resulting in nonlinearities and structural inconsistencies. 
Second, herding often requires balancing various conflicting objectives, such as goal reaching, inter-evader collision avoidance, and evader containment, which complicates the design of a unified control framework. 
Third, partial observability of states presents an additional difficulty, as herders frequently lack access to complete information about the positions and velocities of the evaders, which requires the use of techniques such as state estimation~\cite{zheng2025optimal} or neighbor selection~\cite{li2024collective} to compensate for information loss. 
Finally, scalability and distributed coordination become critical as the number of agents increases, requiring control strategies that are both scalable and decentralized to ensure efficient cooperation while avoiding conflicts and redundancy~\cite{sun2023mean}.

This paper presents a novel cooperative control strategy for coordinating multiple herders to safely herd a group of evaders to a desired goal region using only repulsive interactions and without inducing collisions.
Given the nonlinear, high-order, and underactuated nature of the herding system, control barrier functions(CBFs) are separately constructed to ensure both goal reaching and collision avoidance among the evaders.
A backstepping approach is then employed to recursively design virtual inputs for the control-affine structure system, culminating in a cooperative herding strategy that integrates backstepping CBFs.
The key contributions of this paper are summarized as follows:

 1) We propose a unified cooperative herding framework based on a general tool of CBFs. 
 Leveraging the set-invariance principle, the framework constructs distinct CBFs for the dual objectives of goal reaching and collision avoidance.
 These constraints are seamlessly integrated into a single quadratic program, enabling the systematic synthesis of control inputs that ensure both herding completion and safety guarantees. 

 2) We integrate the backstepping method into the CBF-based herding framework after reformulating the inherently underactuated and nonlinear herding dynamics into a control-affine structure. 
 This hierarchical decomposition enables a recursive control design that avoids taking derivatives of higher-order systems, thereby simplifying controller synthesis. 
 Moreover, the proposed strategy supports both centralized and decentralized implementations.
 
 3) Extensive simulations and hardware experiments have been conducted using mobile robots, demonstrating that herder robots can drive evader robots without collisions into the goal region with minimal movement, relying exclusively on repulsive forces. 
 These results validate the effectiveness of the proposed strategy in safe herding.

The remainder of the paper proceeds as follows. 
Section~\ref{sec2} reviews existing herding strategies and key technical approaches.
Section~\ref{sec3} presents the system model and formulates the herding problem, followed by a brief overview of the CBFs employed in the proposed framework.
In Section~\ref{sec4}, to address the nonlinear and underactuated characteristics of the system, the dynamics are reformulated into a control-affine structure, which enables the effective application of the backstepping method. 
Subsequently, backstepping CBFs are constructed for goal reaching and evader collision avoidance, respectively.
Finally, we propose a unified cooperative herding strategy with both centralized and decentralized implementations to guarantee safe and effective herding. 
Section~\ref{sec5} provides simulation and experimental results to validate the proposed algorithm and offers further analysis.
Section~\ref{sec6} concludes this work by summarizing the main contributions and outlining potential directions for future research to enhance the practicality and scalability of the proposed herding strategy.
% ----------------------------------------------------------------------------------------------
\section{Related Work}\label{sec2}
The herding problem has emerged as a prominent research topic within the domains of multi-robot cooperative systems. 
It encompasses various challenges, such as agent dynamic interaction, spatial coordination, and goal-directed collective behavior. 
Existing literature can be broadly categorized into the following key research directions.
Early efforts focused on understanding the behavioral principles underlying herding phenomena and laid the foundation for designing herding strategies. 
For instance, the work in~\cite{hemelrijk2015underlies} developed simulation frameworks based on the trajectory data collected from real-world interactions between herders and evaders. 
Additionally, a collective perception model for neighbor selection based on a visual attention mechanism was proposed by \cite{li2024collective}, which successfully replicates both the confusion effect and the oddity effect observed in herder-evader interactions, offering insights into how perception and attention influence group behavior.
Quantitative analyses of the relationships between population ratios and velocity correlations have also been explored.
As demonstrated in~\cite{lama2024shepherding, carere2009aerial}, successful herding scenarios critically depend on the proper alignment of the herders' and evaders' population sizes and velocities, indicating that certain configurations are more favorable for effective herding.
More recent studies have shifted toward developing practical strategies for controlling multiple evaders using a single herder, as well as designing cooperative herding algorithms for multiple herder systems~\cite{hamandi2024robotic, sebastian2022adaptive, pierson2017controlling}, which is the primary emphasis of this study.

% single herder
A single herder can theoretically accomplish the herding task~\cite{paranjape2018robotic, singh2024planar}, but this capability comes with a significant trade-off in terms of strategic complexity. 
Early works introduced several typical intuitive forward-steering techniques, such as positioning directly behind or moving side-to-side behind the group~\cite{lien2004shepherding}. 
While these methods are straightforward, they struggle to handle the unpredictable behaviors of evaders, especially when evaders behave evasively or escape the influence zone. 
Lyapunov-based adaptive controllers have been used to stabilize such interactions~\cite{licitra2017single}.
Additionally, leveraging the global center of mass of the evaders can lead to substantial gains in control efficiency~\cite{fujioka2016effective} when full state information is available.
Although a single herder can theoretically herd a group of evaders, the main challenge is developing intelligent and robust control strategies for effective herding. 
To simplify the problem, many existing approaches rely on favorable assumptions, such as full visibility of all evaders, inherent inter-evader repulsion, and high herder mobility.
While these assumptions facilitate algorithm design, they often limit the applicability of these methods in more practical scenarios, highlighting the need for developing cooperative multi-herder frameworks.

% multi-herder
Compared to single-herder scenarios, coordinating multiple herders significantly increases the complexity of control, primarily due to the necessity of maintaining consistent and synchronized strategies across all herders. 
Drawing from multi-robot cooperative control research, various coordination methods have been explored in herding scenarios, including behavior-based approaches~\cite{vo2009behavior}, rule-based strategies~\cite{zhang2024distributed}, and formation-based methods~\cite{song2021herding, pierson2017controlling, kubo2022herd}. 
Beyond these general strategies, more specialized approaches have been developed. 
For instance, the outermost push approach~\cite{zhang2024distributed, zhang2024heterogeneous} achieves simultaneous herding and collecting in a fully distributed manner without the need for global information. 
However, its reliance on potential field formulations introduces design complexities and makes it difficult to guarantee global convergence. 
To handle the highly nonlinear response dynamics of evader motion, implicit control strategies grounded in numerical analysis have been proposed~\cite{sebastian2021multi, sebastian2022adaptive}. 
Instead of explicitly computing control inputs, these methods utilize local feedback to intuitively guide the system toward the desired herding behavior. 
In contrast to herding, defense~\cite{grover2022noncooperative, mohanty2022distributed} focuses on preventing agents from entering a protected region. While CBFs are also used for defense and collision avoidance, their application often introduces higher-order derivatives, substantially increasing computational complexity.

% In addition to the intrinsic complexities of coordinating multiple herders, real-world herding scenarios often encompass a broader spectrum of intertwined tasks, such as pursuit~\cite{chen2023multiple}, encirclement~\cite{pierson2017controlling}, guidance~\cite{kubo2022herd}, and obstacle avoidance~\cite{hamandi2024robotic}. 
% All of these introduce additional layers of control and planning complexity, compounding the challenge of developing robust and generalizable herding strategies.
% % RL 
% Complementing model-based strategies, learning-based approaches have gained increasing attention. 
% Reinforcement learning (RL), in particular, has shown promise in enabling herders to acquire adaptive herding behaviors through trial-and-error interactions with dynamic environments~\cite{zhi2021learning, nguyen2019deep, zhao2025mathematical}.
% Despite its flexibility, RL methods often suffer from sample inefficiency, demand significant computational resources, and require carefully crafted reward functions to ensure safety, stability, and convergence.

Generally, the herding challenges discussed above require control methodologies that guarantee safety, robustness, and scalability in dynamic and uncertain environments. Existing literature on addressing the challenging multiple-herder multiple-evader herding problem includes approaches such as artificial potential fields~\cite{chipade2019herding}, reinforcement learning~\cite{zhi2021learning, nguyen2019deep, zhao2025mathematical}, cascaded-PID~\cite{pierson2017controlling}, and CBFs~\cite{hamandi2024robotic}. Among these, CBFs stand out as a promising solution due to their ability to provide strict safety guarantees and effectively utilize system dynamics, motivating their use in this paper. However, existing solutions struggle to handle the high-order dynamics inherent in the herding problem. Methods such as cascaded-PID and exponential CBF are considered in~\cite{pierson2017controlling,mohanty2022distributed}. In contrast, the backstepping CBF~\cite{taylor2022safe}, which provides strict safety guarantees compared to cascaded-PID and greater adaptability and robustness than the exponential CBF, has not been explored in the existing literature. Backstepping CBF is particularly well-suited for nonlinear, high-dimensional, or multi-agent systems. While the simpler exponential CBF is effective for linear systems with exponential decay, the recursive structure of BCBF allows it to address complex dynamics, uncertainties, and non-convex safety regions, making it more suitable for real-world applications where safety is crucial.

Motivated by the above analysis of existing work, this study focuses on designing a cooperative herding strategy that enables multiple herders to drive a group of evaders into a desired goal region using only repulsive interactions and without inducing collisions. To simultaneously ensure herding performance and safety, we construct separate CBFs for goal reaching and collision avoidance, which are integrated into a unified quadratic programming framework. 
This formulation is readily extendable to incorporate additional task constraints, thereby enhancing the generality and applicability of the proposed framework.
Importantly, to avoid the computational burden associated with higher-order derivatives in the standard CBF approach, we introduce a backstepping-based design after reformulating the system dynamics into a control-affine form. This leads to a cooperative herding strategy using backstepping CBFs, which offers strict safety guarantees, enhanced adaptability, and greater robustness compared to cascaded-PID and exponential CBFs. Backstepping CBFs are especially effective for nonlinear, high-dimensional, and multi-agent systems, making them ideal for real-world applications.

% -------------------------------------------------------------------------------------------------------------------
\section{Problem Formulation and Preliminaries}\label{sec3}
Consider herding a group of $m$ evader robots using $n$ herder robots. 
As shown in Fig.~\ref{Herding_Senario}, the position of the $i$-th evader is $\bf{x}_{\rm{E}_i} \in \mathbb{R}^2,i\in \mathcal{E}:=\{1,...,m\}$, with the collective positions of the evaders denoted as $\bf{x}_\rm{E}:=\left[\bf{x}_{\rm{E}_1}^{\top}, x_{\rm{E}_2}^{\top}, \ldots, \bf{x}_{\rm{E}_m}^{\top}\right]^{\top}$.
Similarly, the position of the $k$-th herder is $\bf{x}_{\rm{H}_k} \in \mathbb{R}^2, k\in\mathcal{K}:=\{1,...,n\}$, with the collective positions of the herders denoted as $ \bf{x}_\rm{H}:= \left[\bf{x}_{\rm{H}_1}^{\top}, \bf{x}_{\rm{H}_2}^{\top}, \ldots, \bf{x}_{\rm{H}_n}^{\top}\right]^{\top}$.  

\begin{figure}[tp]
    \centering
    \makebox[0pt]{%
    \includegraphics[width=0.4\textwidth]{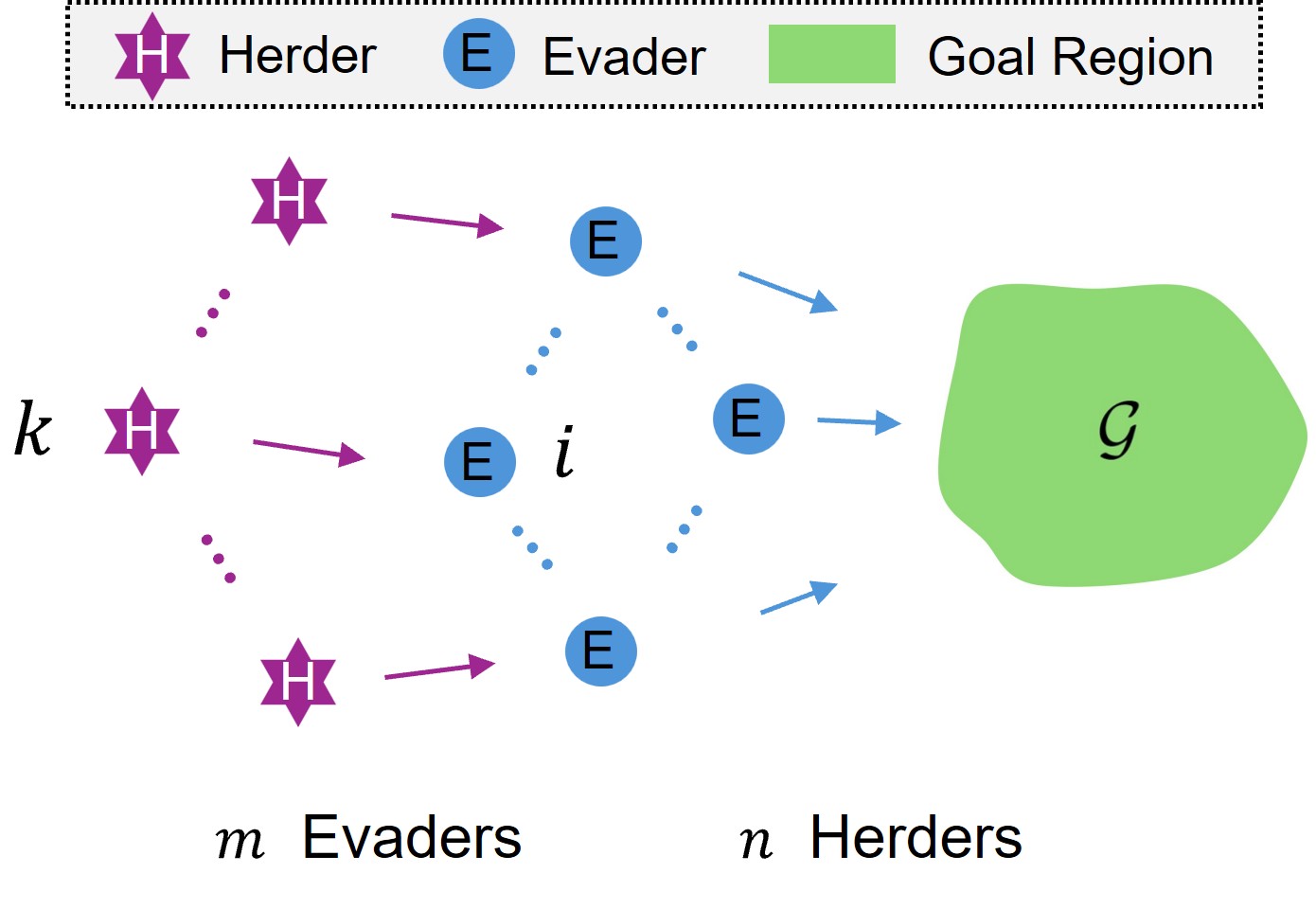}}
    \caption{The settings of the herding problem.}
    \label{Herding_Senario}
\end{figure}

% -------------------------------------------------------------------------------------------------------------------
\subsection{System Model}
For each evader $i \in \mathcal{E}$, we assume its movement is governed by the following dynamics:
\begin{equation}\label{Evader_Dynamics}
\begin{aligned}
\dot{\bf{x}}_{\rm{E}_i} =f_{\rm{E}_i}(\bf{x}_{\rm{E}},\bf{x}_{\rm{H}}),
\end{aligned}
\end{equation}
where $f_{\rm{E}_i}(\cdot)$ is a continuously differentiable nonlinear function. 
As discussed in the literature~\cite{sebastian2022adaptive}, there are two commonly used models for~\eqref{Evader_Dynamics}: the inverse model~\cite{pierson2017controlling} and the exponential model~\cite{licitra2017single}. 
Compared to the exponential model, the inverse model provides smoother and more predictable control due to its gradual force decay, which allows herders to influence evaders effectively even when they are distant. In contrast, the exponential model’s sharp force decay leads to extreme short-range repulsion, which can destabilize simulations and hinder scalability in multi-agent coordination. Furthermore, the inverse model better reflects natural interaction laws and is more compatible with physical robotic control interfaces, making it a more intuitive choice for herding dynamics.

In this study, the inverse model is adopted, where evader velocity is directly determined by the gradient of an artificial potential field generated by the herders, in which the function $f_{\rm{E}_i}(\bf{x}_{\rm{E}}, \bf{x}_{\rm{H}})$ in \eqref{Evader_Dynamics} is given by
\begin{equation}\label{Evader_Dynamics_Specific}
\begin{aligned}
f_{\rm{E}_i}(\bf{x}_{\rm{E}},\bf{x}_{\rm{H}}) =\kappa_{\rm{H}}\sum\limits_{k\in\mathcal{K}}\frac{\bf{x}_{\rm{E}_{i}}-\bf{x}_{\rm{H}_{k}}}{\|\bf{x}_{\rm{E}_{i}}-\bf{x}_{\rm{H}_{k}}\|^{3}},
\end{aligned}
\end{equation}
where $\kappa_{\rm{H}}\in\mathbb{R}_{>0}$ is the proportional gain associated with the repulsive force in the dynamics, indicating that the evaders consistently move away from the herders. 
This approach is widely used in robotics and biological herding models~\cite{pierson2017controlling}, and offers a balance between ease of analysis and realistic modeling. 
Next, the dynamics of the $k$-th herder, where $k \in \mathcal{K}$, are modeled as a single integrator. That is, the herders are velocity-controlled and governed by the following equation: 
\begin{equation}\label{herder_Dynamics}
\begin{aligned}
\dot{\bf{x}}_{\rm{H}_k}&=\bf{u}_{\rm{H}_k},
\end{aligned}
\end{equation}
where $\bf{u}_{\rm{H}_k}$ is the control input of the $k$-th herder. 

An additional point to note is that full-state information is obtained via motion capture, and evader behavior is modeled through known repulsive interactions under an inverse dynamics framework. These assumptions help isolate and validate the core functionality of the proposed backstepping CBF controller. However, in practical scenarios, issues such as sensor noise, partial observability, and model uncertainties may degrade performance and safety. These challenges can be addressed by incorporating state estimation techniques~\cite{bell2016adaptive}, robust/adaptive control methods~\cite{licitra2017single,parikh2019integral}, or learning-based approaches~\cite{licitra2019single}. 
Hence, the following assumption is introduced.
\vspace{0.5em}
\begin{Asum}\label{Knowledge_Dynamics}
The herders are assumed to have full knowledge of the evader dynamics as specified in~\eqref{Evader_Dynamics} and can accurately measure the positions of the evaders, which is a common assumption in the herding literature~\cite{pierson2017controlling, sebastian2022adaptive}. Moreover, they possess precise information about the goal region.
\end{Asum}
\vspace{0.5em}
Given the dynamics of the evaders in~\eqref{Evader_Dynamics} and the herders in~\eqref{herder_Dynamics}, they are assumed to be mass points since this simplification facilitates implementation in both simulations and physical robotic systems. 
Accordingly, the joint system dynamics can be expressed as follows:
\begin{equation}\label{Joint_Dynamics}
    \dot{\bf{x}}=f(\bf{x})+g(\bf{x})\cdot\bf{u}_{\rm{H}},
\end{equation}
where $\bf{x}=[\bf{x}_{\rm{E}}^{\top}, \bf{x}_{\rm{H}}^{\top}]^{\top}$ denotes the overall system state, and $f(\bf{x})=[f_{\rm{E}}(\bf{x})^{\top}, \bf{0}^{\top}]^{\top}$ is the drift vector field with $f_{\rm{E}}(\bf{x})=[f_{\rm{E}_{1}}(\bf{x}),f_{\rm{E}_{2}}(\bf{x}),\cdots,f_{\rm{E}_{m}}(\bf{x})]^{\top}$ representing the evader dynamics. 
The matrix $g(\bf{x})=[\bf{0}^{\top}, \bf{I}^{\top}]^{\top}$ captures the control directions of the system and $\bf{u}_{\rm{H}}=[\bf{u}_{\rm{H}_{1}}^{\top},\bf{u}_{\rm{H}_{2}}^{\top},\cdots,\bf{u}_{\rm{H}_{n}}^{\top}]^{\top}$ denotes the collective control inputs of the herders.

\vspace{0.5em}
\begin{Rmk}\label{remark_underactuated}
When modeling the evaders and herders as heterogeneous agents governed by~\eqref{Joint_Dynamics}, it is important to note that control inputs are applied exclusively to the herders, indirectly influencing the evaders. Consequently, the overall system in~\eqref{Joint_Dynamics} is underactuated since the number of control inputs ($2n$) is insufficient to directly control the full state of the system ($2m + 2n$).
\end{Rmk}

% ------------------------------------------------------------------------------------------------------------
\subsection{Problem Formulation}
This work aims to develop a cooperative control strategy for a team of $n$ herders to collectively drive a group of $m$ evaders into a designated goal region. 
In addition to successful herding, collision-free motion must be ensured, preventing collisions both within and between the groups of evaders and herders.
To further clarify the objective, we provide the following details.

Define the goal region as $\mathcal{G} \subset \mathbb{R}^2$ and assume that it is a circular area with center $\bf{x}_{\text{goal}}$ and radius $R_{\text{goal}}$, given by
\begin{align}
\mathcal{G}:=\left\{\bm{\xi} \in \mathbb{R}^2 \mid\left\|\bm{\xi} -\bf{x}_{\text{goal}}\right\|^{2} \leq R_{\text{goal}}^{2}\right\}.
\end{align}
% where $| \cdot |$ denotes the L2 norm. 
It is assumed that the goal region $\mathcal{G}$ is sufficiently large to accommodate all $m$ evaders.

For the $i$-th evader, $i \in \mathcal{E}$, two types of collisions must be avoided:
i) collisions with the herders. 
ii) collisions with other evaders during the herding process. 
While the first requirement is intrinsically satisfied by the system dynamics as given in~\eqref{Evader_Dynamics} and~\eqref{Evader_Dynamics_Specific}, the second is enforced by imposing the following condition
\begin{equation}
    \mathcal{B}_{ji}=\{\bm{\xi}\in\mathbb{R}^{2}|\|\bf{x}_{\rm{E}_{j}}-\bm{\xi}\|^2\geq R_{\text{avoid}}^{2}\},
\end{equation}
where $R_{\text{avoid}}$ denotes the predefined safe distance between the $j$-th and $i$-th evaders. This can also be interpreted as a safety margin that accounts for the physical size of robots in real-world scenarios. Consequently, we define the set $\mathcal{B}$ to be the set that guarantees collision avoidance among $m$ evaders at all times, given by
\begin{equation}
    \mathcal{B}=\bigcap\limits_{i=1}^{m}\bigcap\limits_{j=i+1}^{m}\mathcal{B}_{ij},
\end{equation}
where $i, j \in\mathcal{E}$, and $i\neq j$.
\vspace{0.5em}
\begin{Rmk}
 In practical scenarios, collision avoidance among herders is also of great importance. 
 However, this work primarily focuses on how the herders' control inputs indirectly influence the behavior of the evaders, which is addressed within the framework of backstepping CBFs. 
 Although collisions between herders can be avoided by incorporating additional CBF constraints based on the herders' dynamics, as given in~\eqref{herder_Dynamics}, such considerations fall outside the scope of the core objectives of this paper. 
 Therefore, the strategy for preventing collisions between herders is not explicitly discussed.
\end{Rmk}

\vspace{0.5em}

\begin{Prob}\label{Herding_Problem}
Consider a herding task $\mathcal{T}$ involving $m$ evaders and $n$ herders governed by~\eqref{Evader_Dynamics} and~\eqref{herder_Dynamics}, respectively. The objective is to design control inputs $\mathbf{u}_{\rm{H}}$ for the herders such that:
1) all evaders reach and remain within a designated goal region $\mathcal{G}$, i.e., $\mathbf{x}_{\rm{E}_i} \in \mathcal{G},\ \forall i \in \mathcal{E}$;
2) inter-evader collision avoidance is guaranteed, i.e., $\bf{x}_{\rm{E}}\in \mathcal{B}$.
\end{Prob}

\subsection{Control Barrier Functions} %  CLFs and CBFs
Specifically, consider a closed set $\mathcal{C}\subset\mathbb{R}^{n}$ defined as the $0$-superlevel set of a continuously differentiable function $h:\mathbb{R}^{n}\rightarrow\mathbb{R}$, given by
\begin{equation}\label{Invariant_Set}
		\begin{aligned}
			\mathcal{C} & \triangleq\left\{\mathbf{x}\in \mathbb{R}^{n}: h(\mathbf{x}) \geq 0\right\}, \\
			\partial \mathcal{C} & \triangleq\left\{\mathbf{x}\in\mathbb{R}^{n}: h(\mathbf{x})=0\right\}, \\
			\operatorname{Int}(\mathcal{C}) & \triangleq\left\{\mathbf{x}\in\mathbb{R}^{n}: h(\mathbf{x})>0\right\},
		\end{aligned}
\end{equation}
where $\operatorname{Int}(\mathcal{C})$ and $\partial \mathcal{C}$ denote the interior and boundary of the set $\mathcal{C}$, respectively. 
\vspace{0.5em}
\begin{Def}
(Extended class-$\mathcal{K}$ function, $\mathcal{K}_{e}$): A continuous function $\alpha:(-b,a)\rightarrow(-\infty,\infty)$, with $a,b>0$, is an extended class $\mathcal{K}$ function ($\alpha\in\mathcal{K}_{e}$) if $\alpha(0)=0$ and $\alpha$ is strictly increasing.
\end{Def}
\vspace{0.5em}
\begin{Def}
(Forward Invariance and Safety ~\cite{ames2016control}): A set $\mathcal{C}\subset\mathbb{R}^{n}$ defined in~\eqref{Invariant_Set} is~\textit{forward invariant} w.r.t.~\eqref{Joint_Dynamics} if for every $\mathbf{x}_{0}\in\mathcal{C}$, the solution $\mathbf{x}(t)$ to~\eqref{Joint_Dynamics} satisfies $\mathbf{x}(t)\in\mathcal{C}$ for all $t\in\bf{T}(\mathbf{x}_{0})$, where $\bf{T}(\mathbf{x}_{0})$ defines a time interval $[0, t_{\max})$,  $t_{\max}\in\mathbb{R}_{>0}$. The system~\eqref{Joint_Dynamics} is~\text{safe} on the set $\mathcal{C}$ if the set $\mathcal{C}$ is forward invariant.
\end{Def}
\vspace{0.5em}
With the aid of the above definitions, the formal definition of CBFs is presented as follows.

\vspace{0.5em}
\begin{Def}\label{CBF_Definition}
		 Let $\mathcal{C}\subset\mathbb{R}^{n}$ be the $0$-superlevel set of a smooth function $h:\mathbb{R}^{n}\rightarrow\mathbb{R}$ defined by \eqref{Invariant_Set}. Then, $h$ is a CBF for system~\eqref{Joint_Dynamics} if there exists a smooth function $\beta\in\mathcal{K}_{e}$ such that, for all $\mathbf{x}\in\mathcal{C}$, there exists a control input $\mathbf{u}\in\mathbb{R}^{m}$ satisfying
		\begin{equation}\label{CBF_Condition}
    \begin{split}
        &L_{\mathbf{f}} h(\mathbf{x})+L_{\mathbf{g}} h(\mathbf{x})\mathbf{u}\geq -\beta({h(\mathbf{x})}),
    \end{split}
\end{equation}
where $\beta(\cdot)$ is an extended class-$\mathcal{K}$ function.
\end{Def}

% -------------------------------------------------------------------------------------------------------------
\section{Controller Design}\label{sec4}
In this section, we propose a cooperative herding strategy based on backstepping CBFs, enabling multiple herders to collaboratively steer a group of evaders safely toward the target region solely through repulsive interactions.

% -------------------------------------------------------------------------------------------------------------
\subsection{Reformulating the Dynamical Model}
To facilitate the design of backstepping controllers for the dynamical model~\eqref{Evader_Dynamics} and~\eqref{herder_Dynamics}, the state $\mathbf{x}_{\mathrm{H}}$ in~\eqref{Evader_Dynamics} will first be treated as a virtual control input. However, this treatment makes the dynamical model~\eqref{Evader_Dynamics} not control-affine. Consequently, the results presented in~\cite{taylor2022safe}, which require a control-affine structure for the dynamical models in~\eqref{Evader_Dynamics} and~\eqref{herder_Dynamics}, cannot be applied directly.

To address this issue, we reformulate the system dynamics~\eqref{Evader_Dynamics} and~\eqref{herder_Dynamics}. Specifically, define an intermediate variable $  \mathbf{v}_{\mathrm{E}_{i}} := f_{\rm{E}_i}(\mathbf{x}_{\rm{E}}, \mathbf{x}_{\rm{H}})$, where the function $f_{\rm{E}_i}(\mathbf{x}_{\rm{E}}, \mathbf{x}_{\rm{H}})$ is provided in~\eqref{Evader_Dynamics_Specific}. 
With this definition, the evader dynamics~\eqref{Evader_Dynamics} can be equivalently expressed as
\begin{equation}\label{Evader_dynamics_Reformulate}
    \dot{\bf{x}}_{\rm{E}_i} = \mathbf{v}_{\mathrm{E}_{i}}.
\end{equation}
Next, based on the expressions in~\eqref{Evader_Dynamics_Specific} and~\eqref{herder_Dynamics}, the time derivative of $\mathbf{v}_{\mathrm{E}_{i}}$ can be derived as
\begin{equation}\label{Herder_dynamics_Reformulate}
\begin{split}
   \dot{\mathbf{v}}_{\mathrm{E}_{i}} = \frac{\partial f_{{\rm{E}}_{i}}}{\partial \bf{x}_{{\rm{E}}_{i}}} \mathbf{v}_{\mathrm{E}_{i}} +\frac{\partial f_{{\rm{E}}_{i}}}{\partial \bf{x}_{{\rm{H}}_{k}}} \bf{u}_{{\rm{H}}_{k}},
\end{split}
\end{equation}
where
\begin{equation*}
\begin{aligned}
\frac{\partial f_{{\rm{E}}_{i}}}{\partial \bf{x}_{{\rm{E}}_{i}}}&=\kappa_{\rm{H}} \sum_{k \in \mathcal{K}}\left(\frac{\mathbf{I}_{2\times 2}}{\left\|\bf{d}_{i,k} \right\|^3}-3 \frac{\bf{d}_{i,k} \bf{d}_{i,k} ^{\top}}{\left\|\bf{d}_{i,k} \right\|^5}\right),\\
\frac{\partial f_{{\rm{E}}_{i}}}{\partial \bf{x}_{{\rm{H}}_{k}}}&=-\frac{\partial f_{{\rm{E}}_{i}}}{\partial \bf{x}_{{\rm{E}}_{i}}},\quad \bf{d}_{i,k} = \bf{x}_{{\rm{E}}_{i}} - \bf{x}_{{\rm{H}}_{k}},
% &\frac{\partial f_{{\rm{E}}_{i}}}{\partial \bf{x}_{{\rm{H}}_{k}}}=-\kappa_{\rm{H}} \sum_{k \in \mathcal{K}}\left(\frac{1}{\left\|\bf{d}_{i,k} \right\|^3}-3 \frac{\bf{d}_{i,k} \bf{d}_{i,k} ^{\top}}{\left\|\bf{d}_{i,k} \right\|^5}\right).
\end{aligned}
\end{equation*}
and $\mathbf{I}_{2\times 2}\in  \mathbb{R}^2$ is the identity matrix. As shown in equations~\eqref{Evader_dynamics_Reformulate} and~\eqref{Herder_dynamics_Reformulate}, the control inputs for the reformulated systems are $\bf{v}_{{\rm{E}_i}}$ and $\bf{u}_{\rm{H}_{k}}$, respectively. Since both dynamical models now exhibit a control-affine structure, the backstepping framework proposed in~\cite{taylor2022safe} becomes applicable for addressing the herding task.

% -------------------------------------------------------------------------------------------------------------
\subsection{Backstepping Control Barrier Functions}
\subsubsection{Goal Reaching}
For Problem~\ref{Herding_Problem}, one might consider employing backstepping control Lyapunov functions to address the goal-reaching objective. 
However, since the aim is not to converge to a specific point but rather to ensure that all evaders are driven into and remain within the goal region $\mathcal{G}$, it is more appropriate to utilize a backstepping CBF. 
To this end, a candidate CBF is constructed as follows:
\begin{equation}\label{CLF_Candidate}
    h_1(\bf{x}_{\rm{E}_{i}}) =R_{\text{goal}}^{2}- \|\bf{x}_{\rm{E}_{i}} - \bf{x}_{\text{goal}}\|^{2} , \quad\forall i\in\mathcal{E}.
\end{equation}

When \(h_1(\mathbf{x}_{\rm{E}_i}) < 0\), indicating that the evader is outside the goal region, the CBF condition defined in~\eqref{CBF_Definition} guarantees the convergence of the system state \(\mathbf{x}_{\rm{E}_i}\) to the goal region \(\mathcal{G}\) (see~\cite{xu2015robustness}). 
Specifically, by defining the CLF as \(V_1(\mathbf{x}_{\rm{E}_i}) = -h_1(\mathbf{x}_{\rm{E}_i})\), the CBF condition will reduce to a CLF condition that ensures the convergence of \(\mathbf{x}_{\rm{E}_i}\) to the goal region. 
CBFs provide a more general and theoretically rigorous framework for simultaneously ensuring convergence (when outside the goal set) and invariance (once within it), which is particularly important for the safety-critical nature of the herding task considered in this work.

Once the evader enters the goal region, the same CBF condition ensures the set invariance property, thereby
addressing the goal-reaching task in Problem~\ref{Herding_Problem}.
Therefore, based on the definition of the CBF in Definition~\ref{CBF_Definition}, the following CBF condition can be derived for the dynamical system~\eqref{Evader_dynamics_Reformulate}: 
\begin{equation}\label{CLF_First_Order}
    -2\left(\bf{x}_{\rm{E}_{i}} - \bf{x}_{\text{goal}}\right)^{\top}\cdot\mathbf{v}_{\mathrm{E}_{i}}+\gamma_h h_1(\bf{x}_{\rm{E}_{i}}) \geq 0,
\end{equation}
where $\gamma_{h}>0$ is a constant.

To meet the backstepping control design requirements, the virtual controller $\mathbf{v}_{\mathrm{E}_{i}}=\bf{r}_h(\bf{x}_{\rm{E}_{i}})$ must be designed to satisfy the CBF condition specified in~\eqref{CLF_First_Order}. 
Furthermore, it is essential to ensure that $\mathbf{r}_h(\mathbf{x}_{\rm{E}_{i}})$ is smooth, as discussed in~\cite{taylor2022safe}. 
This can typically be achieved using Sontag's universal formula~\cite{sontag1989universal} or the smooth universal formula~\cite{li2024unifying}. 
However, due to the simplicity of the single integrator model~\eqref{Evader_dynamics_Reformulate}, it can be directly verified that the following control law satisfies the CBF condition in~\eqref{CLF_First_Order}: 
\begin{equation}\label{first_virtual_reaching}
    \bf{r}_h(\bf{x}_{\rm{E}_{i}})=-\gamma_{h}\left(\bf{x}_{\rm{E}_{i}}-\bf{x}_{\text{goal}}\right).
\end{equation}
Note that the control law \(\mathbf{r}_h(\mathbf{x}_{\rm{E}_i})\) can easily be verified to satisfy~\eqref{CLF_First_Order} by substituting~\eqref{first_virtual_reaching} into the condition.

% Next, following the approach in~\cite{taylor2022safe}, we construct the following new candidate CBF for the dynamical model~\eqref{Herder_dynamics_Reformulate}:
Next, following the approach in~\cite{taylor2022safe}, we construct the following backstepping CBF for the dynamical model~\eqref{Herder_dynamics_Reformulate}:
\begin{align}
    h_1(\mathbf{x}_{\rm{E}_{i}}, \mathbf{v}_{\rm{E}_{i}}) &= R_{\text{goal}}^2 - \|\mathbf{x}_{\rm{E}_{i}} - \mathbf{x}_{\text{goal}}\|^2  \notag \\
    &\quad - \frac{1}{2\mu} \|\mathbf{v}_{\rm{E}_{i}} - \mathbf{r}_{h}(\mathbf{x}_{\rm{E}_{i}})\|^2, \label{sencond_candidate_Lyapunov}
\end{align}
where $\mu\in\mathbb{R}_{>0}$. 

Subsequently, to construct a CBF that satisfies the goal reaching condition with $h_1(\bf{x}_{\rm{E}_{i}},\bf{v}_{\rm{E}_{i}})$ given in~\eqref{sencond_candidate_Lyapunov}, we adopt the following structure based on Definition~\ref{CBF_Definition}:

% the following CBF condition is derived based on Definition~\ref{CBF_Definition}:
% \begin{equation}\label{Final_Lyapunov_Condition}
% \begin{aligned}
% 			% \dot{V}(\bf{x}_{\rm{E}_{i}},\bf{v}_{\rm{E}_{i}},\bf{u}_{{\rm{H}}_{k}})=&
%         &-2(\bf{x}_{\rm{E}_{i}} - \bf{x}_{\text{goal}})^{\top} \bf{v}_{\rm{E}_{i}} - \frac{1}{\mu}(\bf{v}_{\rm{E}_{i}}-\bf{r}_{h}(\bf{x}_{\rm{E}_{i}}))^{\top}\Big(\frac{\partial f_{{\rm{E}}_{i}}}{\partial \bf{x}_{{\rm{E}}_{i}}}\cdot \mathbf{v}_{\mathrm{E}_{i}} \\
%      &+\frac{\partial f_{{\rm{E}}_{i}}}{\partial \bf{x}_{{\rm{H}}_{k}}}\cdot \bf{u}_{{\rm{H}}_{k}}+\gamma_h \mathbf{v}_{\mathrm{E}_{i}} \Big)+\gamma_h h_1(\bf{x}_{\rm{E}_{i}},\bf{v}_{\rm{E}_{i}})\geq 0.
%      % \label{Final_Lyapunov_Condition}
% \end{aligned}
% \end{equation}
\begin{equation}\label{Final_Lyapunov_Condition}
\begin{aligned}
			% \dot{V}(\bf{x}_{\rm{E}_{i}},\bf{v}_{\rm{E}_{i}},\bf{u}_{{\rm{H}}_{k}})=&
        &-2(\bf{x}_{\rm{E}_{i}} - \bf{x}_{\text{goal}})^{\top} \bf{v}_{\rm{E}_{i}} - \frac{1}{\mu}(\bf{v}_{\rm{E}_{i}}-\bf{r}_{h}(\bf{x}_{\rm{E}_{i}}))^{\top}  \\
        &\cdot\Big(\frac{\partial f_{{\rm{E}}_{i}}}{\partial \bf{x}_{{\rm{E}}_{i}}}\cdot \mathbf{v}_{\mathrm{E}_{i}}  
      +\frac{\partial f_{{\rm{E}}_{i}}}{\partial \bf{x}_{{\rm{H}}_{k}}}\cdot \bf{u}_{{\rm{H}}_{k}}+\gamma_h \mathbf{v}_{\mathrm{E}_{i}} \Big)\\
      &+\gamma_h h_1(\bf{x}_{\rm{E}_{i}},\bf{v}_{\rm{E}_{i}})\geq 0.
     % \label{Final_Lyapunov_Condition}
\end{aligned}
\end{equation}

Next, based on~\eqref{Final_Lyapunov_Condition}, Sontag's universal formula is employed to design the stabilizing control law. 
To achieve this, we define the following variables:
\begin{equation}\label{Sontag_variable}
\begin{split}
    &a(\bf{x}_{\rm{E}_{i}}, \bf{v}_{\rm{E}_{i}}):=-2(\bf{x}_{\rm{E}_{i}} - \bf{x}_{\text{goal}})^{\top}\bf{v}_{\rm{E}_{i}}+\gamma_h h_1(\bf{x}_{\rm{E}_{i}},\bf{v}_{\rm{E}_{i}})\\
    & 
     -\frac{1}{\mu}(\bf{v}_{\rm{E}_{i}}-\bf{r}_{h}(\bf{x}_{\rm{E}_{i}}))^{\top}\Big(\frac{\partial f_{{\rm{E}}_{i}}}{\partial \bf{x}_{{\rm{E}}_{i}}}\cdot \mathbf{v}_{\mathrm{E}_{i}}+\gamma_h\mathbf{v}_{\mathrm{E}_{i}}\Big),\\
     &\mathbf{b}(\bf{x}_{\rm{E}_{i}},\bf{v}_{\rm{E}_{i}}):=-\frac{1}{\mu}(\bf{v}_{\rm{E}_{i}}-\bf{r}_{h}(\bf{x}_{\rm{E}_{i}}))^{\top}\cdot \frac{\partial f_{{\rm{E}}_{i}}}{\partial \bf{x}_{{\rm{H}}_{k}}}.
\end{split}
\end{equation}
Consequently, condition~\eqref{Final_Lyapunov_Condition} can be equivalently rewritten as follows:
\begin{equation}\label{CLF_Herding_Reformulate}
    a(\bf{x}_{\rm{E}_{i}}, \bf{v}_{\rm{E}_{i}})+\mathbf{b}(\bf{x}_{\rm{E}_{i}},\bf{v}_{\rm{E}_{i}})\cdot\bf{u}_{{\rm{H}}_{k}}\geq 0.
\end{equation}

\setcounter{figure}{1}
\begin{figure*}[!http] % 星号表示跨两栏，通常在页顶显示
    \centering
    \begin{subfigure}[b]{1\textwidth}
        \includegraphics[width=\textwidth]{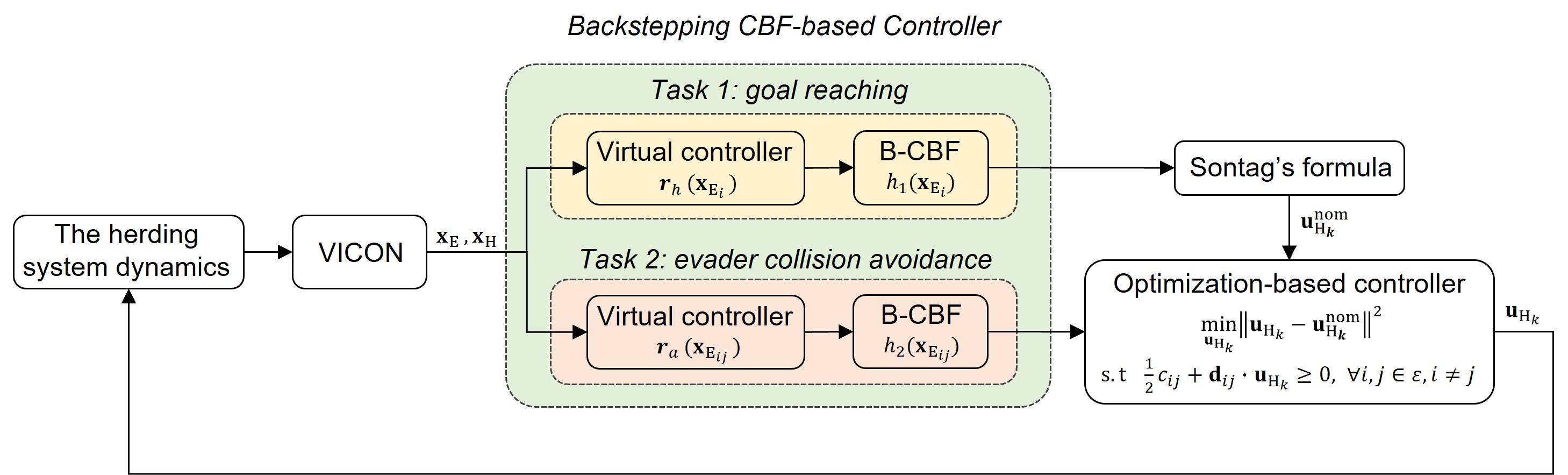}
        % \caption{initial configuration}
        % \label{herding_1}
    \end{subfigure}
    \hfill
    \caption{Architecture of the cooperative safe herding strategy.}
    \label{diagram}
\end{figure*}

% \begin{figure}[tp]
%     \centering
%     \makebox[0pt]{%
%     \includegraphics[width=0.4\textwidth]{fig/herding_setting.jpg}}
%     \caption{The settings of the herding problem.}
%     \label{Herding_Senario}
% \end{figure}

% -------------------------------------------------------------------------------------------------------------
% \subsection{Backstepping Control Barrier Functions}
\subsubsection{Evader Collision Avoidance}
The next requirement is to ensure that evaders avoid collisions with each other during the herding process. To address this, we define the following candidate CBF:
\begin{equation}\label{CBF_Candidate}
    h_2(\bf{x}_{\rm{E}_{ij}}) = \|\bf{x}_{\rm{E}_{ij}}\|^{2} - R_{\text{avoid}}^{2}, \forall i, j \in\mathcal{E}, i\neq j ,
\end{equation}
where $\bf{x}_{\rm{E}_{ij}}=\bf{x}_{\rm{E}_{i}} - \bf{x}_{\rm{E}_{j}}$, and $R_{\text{avoid}}$ represents the safety distance between evaders. 
Similarly to the design of the goal-reaching controller, we first apply the CBF condition given in~\eqref{CBF_Condition} to design a virtual controller $\bf{v}_{\rm{E}_{ij}} = \mathbf{r}_{a}(\bf{x}_{\rm{E}_{ij}})$ for the dynamical model~\eqref{Evader_dynamics_Reformulate}. Then, using the candidate CBF presented in~\eqref{CBF_Candidate}, the following condition is obtained:
\begin{equation}\label{CBF_Synthesis}
    2\bf{x}_{\rm{E}_{ij}}^{\top} \cdot\bf{v}_{\rm{E}_{ij}} \geq  -\gamma _ah_2(\bf{x}_{\rm{E}_{ij}}),
\end{equation}
where $\bf{v}_{\rm{E}_{ij}}=\bf{v}_{\rm{E}_{i}}-\bf{v}_{\rm{E}_{j}}$, $\gamma_a\in\mathbb{R}_{>0}$ and $\gamma _ah_2(\bf{x}_{\rm{E}_{ij}})$ is served as a specific $\mathcal{K}_{e}$ function as given in~\eqref{CBF_Condition}. According to the demonstration in~\cite{taylor2022safe}, a smooth controller $\bf{v}_{\rm{E}_{ij}}=\mathbf{r}_{a}(\bf{x}_{\rm{E}_{ij}})$ is designed based on~\eqref{CBF_Synthesis}. 
To achieve this, $\mathbf{r}_{a}(\bf{x}_{\rm{E}_{ij}})$ is defined as follows:
\begin{equation}\label{Safe_Virtual}
    \mathbf{r}_{a}(\bf{x}_{\rm{E}_{ij}})=\gamma_a \cdot\bf{x}_{\rm{E}_{ij}}.
\end{equation}
The controller provided in~\eqref{Safe_Virtual} can be easily verified to satisfy the condition in~\eqref{CBF_Synthesis}. Next, for the dynamical model~\eqref{Herder_dynamics_Reformulate}, we introduce the following backstepping CBF:
% \begin{align}
%     h_2(\bf{x}_{\rm{E}_{ij}},\bf{v}_{\rm{E}_{ij}})=&\|\bf{x}_{\rm{E}_{ij}}\|^{2}- R_{\text{avoid}}^{2} \notag\\
%     &-\frac{1}{2\mu}\|\bf{v}_{\rm{E}_{ij}}-\mathbf{r}_{a}(\bf{x}_{\rm{E}_{ij}})\|^{2}.
%     \label{sencond_candidate_Barrier}
% \end{align}
\begin{equation}\label{sencond_candidate_Barrier}
\begin{split}
    h_2(\bf{x}_{\rm{E}_{ij}},\bf{v}_{\rm{E}_{ij}})=&\|\bf{x}_{\rm{E}_{ij}}\|^{2}- R_{\text{avoid}}^{2} \notag\\
    &-\frac{1}{2\mu}\|\bf{v}_{\rm{E}_{ij}}-\mathbf{r}_{a}(\bf{x}_{\rm{E}_{ij}})\|^{2}.
\end{split}
\end{equation}
Subsequently, to construct a CBF that satisfies the evader collision avoidance condition with $h_2(\bf{x}_{\rm{E}_{ij}},\bf{v}_{\rm{E}_{ij}})$ given in~\eqref{sencond_candidate_Barrier}, we adopt the following structure based on Definition~\ref{CBF_Definition}:

% with $h_2(\bf{x}_{\rm{E}_{ij}},\bf{v}_{\rm{E}_{ij}})$ given in~\eqref{sencond_candidate_Barrier}, the following CBF condition is derived based on Definition~\ref{CBF_Definition}:
% \begin{equation}\label{CBF_Herding}
% \begin{split}
%     &2\bf{x}_{\rm{E}_{ij}}^{\top}\bf{v}_{\rm{E}_{ij}} -\frac{1}{\mu}(\bf{v}_{\rm{E}_{ij}}-\mathbf{r}_{a}(\bf{x}_{\rm{E}_{ij}}))^{\top} \left( \left( \frac{\partial f_{{\rm{E}}}}{\partial \mathbf{x}_{{\rm{E}}}} \right)_{ij}\cdot \mathbf{v}_{\mathrm{E}_{ij}} \right.\\
%      &\left. +\left( \frac{\partial f_{{\rm{E}}}}{\partial \mathbf{x}_{{\rm{H}}}} \right)_{kq}\cdot \bf{u}_{{\rm{H}}_{kq}}-\gamma_a\mathbf{v}_{\mathrm{E}_{ij}} \right)  +\gamma_a h_2(\bf{x}_{\rm{E}_{ij}},\bf{v}_{\rm{E}_{ij}})\geq 0,
% \end{split}
% \end{equation}
\begin{equation}\label{CBF_Herding}
\begin{split}
    &2\bf{x}_{\rm{E}_{ij}}^{\top}\bf{v}_{\rm{E}_{ij}} -\frac{1}{\mu}(\bf{v}_{\rm{E}_{ij}}-\mathbf{r}_{a}(\bf{x}_{\rm{E}_{ij}}))^{\top} \\
    &\cdot \bigg( \left( \frac{\partial f_{{\rm{E}}}}{\partial \mathbf{x}_{{\rm{E}}}} \right)_{ij}\cdot \mathbf{v}_{\mathrm{E}_{ij}}  +\left( \frac{\partial f_{{\rm{E}}}}{\partial \mathbf{x}_{{\rm{H}}}} \right)_{kq}\cdot \bf{u}_{{\rm{H}}_{kq}}-\gamma_a\mathbf{v}_{\mathrm{E}_{ij}} \bigg)  \\
     &+\gamma_a h_2(\bf{x}_{\rm{E}_{ij}},\bf{v}_{\rm{E}_{ij}})\geq 0,
\end{split}
\end{equation}
 where $ \left( \frac{\partial f_{{\rm{E}}}}{\partial \mathbf{x}_{{\rm{E}}}} \right)_{ij}=\frac{\partial f_{{\rm{E}}_{i}}}{\partial \bf{x}_{{\rm{E}}_{i}}}-\frac{\partial f_{{\rm{E}}_{j}}}{\partial \bf{x}_{{\rm{E}}_{j}}}$, $\left( \frac{\partial f_{{\rm{E}}}}{\partial \mathbf{x}_{{\rm{H}}}} \right)_{kq}=\frac{\partial f_{{\rm{E}}_{i}}}{\partial \bf{x}_{{\rm{H}}_{k}}}-\frac{\partial f_{{\rm{E}}_{j}}}{\partial \bf{x}_{{\rm{H}}_{q}}}$, and $\bf{u}_{{\rm{H}}_{kq}}=\bf{u}_{{\rm{H}}_{k}}-\bf{u}_{{\rm{H}}_{q}}$. Afterward, the same process as in~\eqref{Sontag_variable} is followed, and we define:
\begin{equation}\label{Sontag_CBF_Variable}
    \begin{aligned}
        &c(\mathbf{x}_{\rm{E}_{ij}},\mathbf{v}_{\rm{E}_{ij}}) 
        := 2\mathbf{x}_{\rm{E}_{ij}}^{\top} \mathbf{v}_{\rm{E}_{ij}} 
        + \gamma_a h_2(\mathbf{x}_{\rm{E}_{ij}}, \mathbf{v}_{\rm{E}_{ij}}) \\
        &- \frac{1}{\mu} (\mathbf{v}_{\rm{E}_{ij}} - \mathbf{r}_{a}(\mathbf{x}_{\rm{E}_{ij}}))^{\top} 
        \bigg( \left( \frac{\partial f_{\rm{E}}}{\partial \mathbf{x}_{\rm{E}}} \right)_{ij} 
        \mathbf{v}_{\rm{E}_{ij}} - \gamma_a \mathbf{v}_{\rm{E}_{ij}} \bigg), \\[6pt]
        &\mathbf{d}(\mathbf{x}_{\rm{E}_{ij}},\mathbf{v}_{\rm{E}_{ij}}) 
        := -\frac{1}{\mu} (\mathbf{v}_{\rm{E}_{ij}} - \mathbf{r}_{a}(\mathbf{x}_{\rm{E}_{ij}}))^{\top} 
        \left( \frac{\partial f_{\rm{E}}}{\partial \mathbf{x}_{\rm{H}}} \right)_{kq}.
    \end{aligned}
\end{equation}
where $k,q\in\mathcal{K}, k\neq q$.
Therefore, the condition~\eqref{CBF_Herding} can be equivalently rewritten in the following form: 
\begin{equation}\label{CBF_Herding_Reformulate}
    c(\bf{x}_{\rm{E}_{ij}},\bf{v}_{\rm{E}_{ij}})+\mathbf{d}(\bf{x}_{\rm{E}_{ij}},\bf{v}_{\rm{E}_{ij}})\cdot\bf{u}_{{\rm{H}}_{kq}}\geq 0.
\end{equation}
However, it is important to note that condition~\eqref{CBF_Herding_Reformulate} is only applicable for centralized implementation. 
In contrast, for a distributed implementation, it is necessary to design controllers for herders $k$ and $q$ based on the condition~\eqref{CBF_Herding_Reformulate}. 
Specifically, following the approach presented in~\cite{wang2017safety}, condition~\eqref{CBF_Herding_Reformulate} can be decomposed into the following separate conditions for herders $k$ and $q$, respectively:
\begin{subequations}\label{Splitted_CBF_Condition}
    \begin{align}
        &\frac{1}{2}c(\mathbf{x}_{\rm{E}_{ij}},\mathbf{v}_{\rm{E}_{ij}}) + \mathbf{d}(\mathbf{x}_{\rm{E}_{ij}},\mathbf{v}_{\rm{E}_{ij}})\cdot\mathbf{u}_{{\rm{H}}_{k}} \geq 0, \label{Splitted_CBF_a} \\
        &\frac{1}{2}c(\mathbf{x}_{\rm{E}_{ij}},\mathbf{v}_{\rm{E}_{ij}}) - \mathbf{d}(\mathbf{x}_{\rm{E}_{ij}},\mathbf{v}_{\rm{E}_{ij}})\cdot\mathbf{u}_{{\rm{H}}_{q}} \geq 0. \label{Splitted_CBF_b}
    \end{align}
\end{subequations}
Generally, condition~\eqref{Splitted_CBF_Condition} serves as a sufficient condition for~\eqref{CBF_Herding_Reformulate}, where the weighting coefficient $\frac{1}{2}$ in the original formulation can be generalized. 
Specifically, the two inequalities in~\eqref{Splitted_CBF_Condition} can be weighted by $ \alpha $ and $ \beta $, respectively, where $ 0<\alpha<1$, $0<\beta<1$, and $\alpha+\beta=1$. 
This generalized weighting formulation allows for flexible distribution of responsibility between the two constraints while preserving the overall sufficiency of the condition.

\subsection{Decentralized Controllers}
To address Problem~\ref{Herding_Problem}, controllers need to be designed for each herder to ensure that all evaders reach the goal region $\mathcal{G}$ while avoiding collisions. 
This requires satisfying the CBF conditions ~\eqref{CLF_Herding_Reformulate} and~\eqref{Splitted_CBF_Condition}. 
To achieve this, we adopt the safety filter approach proposed in~\cite{ames2019control}, wherein a nominal controller is first constructed to facilitate the stabilizing control task of guiding evaders to the goal region $\mathcal{G}$. 
This nominal controller is formulated based on Sontag’s universal formula~\cite{sontag1989universal}, which provides a smooth feedback law.
The smoothness property is essential, as it ensures the continuity and well-posedness of the overall control input, enabling stable integration with the CBF-based safety constraints in the subsequent quadratic program formulation. 
Therefore, we can obtain
\begin{equation}\label{Sontag_formula}
    \mathbf{u}_{\rm{H}_{k}}^{\text{nom}}=\frac{-a+\sqrt{a^{2}+\|\mathbf{b}\|^{4}}}{\mathbf{b}\mathbf{b}^{\top}}\cdot\mathbf{b}^{\top},
\end{equation}
where $a$ and $\mathbf{b}$ correspond to the variable $a(\bf{x}_{\rm{E}_{i}},\bf{v}_{\rm{E}_{i}})$ and $\mathbf{b}(\bf{x}_{\rm{E}_{i}},\bf{v}_{\rm{E}_{i}})$ from~\eqref{Sontag_variable} for notation simplification. 
Next, to implement the safety filter for the $k$-th herder, the following optimization problem is formulated:
\begin{equation}\label{controller_kth_herder}
\begin{split}
    &\min_{\mathbf{u}_{\rm{H}_k}\in\mathbb{R}^{n}} \quad   \|\mathbf{u}_{\rm{H}_{k}}-\mathbf{u}_{\rm{H}_{k}}^{\text{nom}}\|^{2} \\
    \text{s.t.} \quad  \frac{1}{2}c_{ij} & + \mathbf{d}_{ij}\cdot\mathbf{u}_{{\rm{H}}_{k}} \geq 0, \,\,\,\forall i,j\in\mathcal{E}, i\neq j,
\end{split}
\end{equation}
where $c_{ij}$ and $\mathbf{d}_{ij}$ denote $c(\mathbf{x}_{\rm{E}_{ij}},\mathbf{v}_{\rm{E}_{ij}})$ and $\mathbf{d}(\mathbf{x}_{\rm{E}_{ij}},\mathbf{v}_{\rm{E}_{ij}})$ given in~\eqref{Sontag_CBF_Variable}, respectively. 
Notably, the algorithm assumes comprehensive knowledge of the evaders' dynamics and full observability of all evaders' positions. 
The overall architecture of our proposed herding strategy is shown in Fig.~\ref{diagram} and summarized in Algorithm~\ref{Algorithm1}.

% \textcolor{blue}{
% In the above decentralized controllers, the goal-reaching task for evaders is guaranteed by constraint~\eqref{CLF_Herding_Reformulate}, which drives the evaders to converge to and remain within the goal region through the control input of the herders. Additionally, collision avoidance among evaders is achieved through condition~\eqref{CBF_Herding_Reformulate}, ensuring a safe separation between all evaders. Based on Sontag’s formula, a smooth nominal input $\mathbf{u}_{\rm{H}_{k}}^{\text{nom}}$ is designed to satisfy the goal-reaching constraint, and the final herder control input is generated by solving a constrained optimization problem. 
% }
% % The proposed algorithm is summarized in Algorithm~\ref{Algorithm1}.
% \textcolor{blue}{The overall architecture of our proposed herding strategy is shown in Fig.~\ref{diagram} and summarized in Algorithm~\ref{Algorithm1}.}

% \vspace{1em}
% \begin{Rmk}
In the optimization problem defined by~\eqref{controller_kth_herder}, we consider the scenario where the $i$-th evader avoids collisions with all the $j$-th evaders, where $i,j \in \mathcal{E}$ and $j \neq i$. 
As the number of evaders increases, the computational and sensing demands grow significantly. 
To address this, a neighborhood strategy from~\cite{wang2017safety} is employed, focusing only on pairs of nearby evaders. 
Specifically, the neighborhood set $\mathcal{N}$ is defined to include only those evaders within a predefined constant distance $D$, as expressed in the following equation:
\begin{equation}
    \mathcal{N}:=\{j\in\mathcal{E}|\|\mathbf{x}_{\rm{E}_{i}}-\mathbf{x}_{\rm{E}_{j}}\|\leq D\}.
\end{equation} 
Finally, the neighborhood strategy simplifies the optimization problem in~\eqref{controller_kth_herder} by replacing $j \in \mathcal{E}$ with $j \in \mathcal{N}$, thus restricting the consideration to only nearby evaders.

\vspace{3mm} 
\begin{Thm}\label{Summary_Theorem}
Assume that $h_1(\mathbf{x}_{\mathrm{E}_{i}})$ defined in~\eqref{CLF_Candidate} and $h_2(\mathbf{x}_{\mathrm{E}_{ij}},\mathbf{v}_{\mathrm{E}_{ij}})$ defined in~\eqref{sencond_candidate_Barrier} are continuously differentiable functions. Use Sontag's formula to design stabilizing controllers $\mathbf{u}_{\mathrm{H}_{k}}^{\text{nom}}$ for all evaders according to~\eqref{Sontag_formula} and serve as a nominal controller for a safety filter. Then we design the safety-critical controller defined in~\eqref{controller_kth_herder} to address the herding task. If $c_{ij}>0$, collision-free herding can be ensured, i.e., $\bf{x}_{\rm{E}}\in \mathcal{B}$.
\end{Thm}
\begin{proof}
First, since $c_{ij} \geq 0,\, \forall i,j \in \mathcal{E}, i \neq j$ when $\mathbf{d}_{ij} = 0$ is satisfied by assumption, the optimization problem in Eq.~\eqref{controller_kth_herder} is always feasible. Then, since the resulting controller $\mathbf{u}_{\rm{H}_k}$ is the solution to the optimization~\eqref{controller_kth_herder}, the constraints $\frac{1}{2}c_{ij} + \mathbf{d}_{ij} \cdot \mathbf{u}_{{\rm{H}}_{k}} \geq 0$ are satisfied for all $i,j \in \mathcal{E}, i \neq j$. In this case, we can always ensure the satisfaction of Eq.~\eqref{CBF_Herding_Reformulate} because Eq.~\eqref{Splitted_CBF_Condition} is a sufficient condition for Eq.~\eqref{CBF_Herding_Reformulate}. Therefore, the controller $\mathbf{u}_{\rm{H}_k}$   always exists and can be designed to satisfy all constraints, ensuring the safety of the system.
\end{proof}

Note that in an environment densely populated with herders and evaders, the safety conditions in Eq.~\eqref{controller_kth_herder} may conflict, which violates the assumption \( c_{ij} \geq 0 \) for some \( i, j \) in Theorem~\ref{Summary_Theorem}. In such cases, relaxing certain constraints could compromise safety guarantees, but can still guarantee the feasibility of the optimization~\eqref{controller_kth_herder}.

\setcounter{figure}{2}
\begin{figure*}[!http] % 星号表示跨两栏，通常在页顶显示
    \centering
    \begin{subfigure}[b]{0.32\textwidth}
        \includegraphics[width=\textwidth]{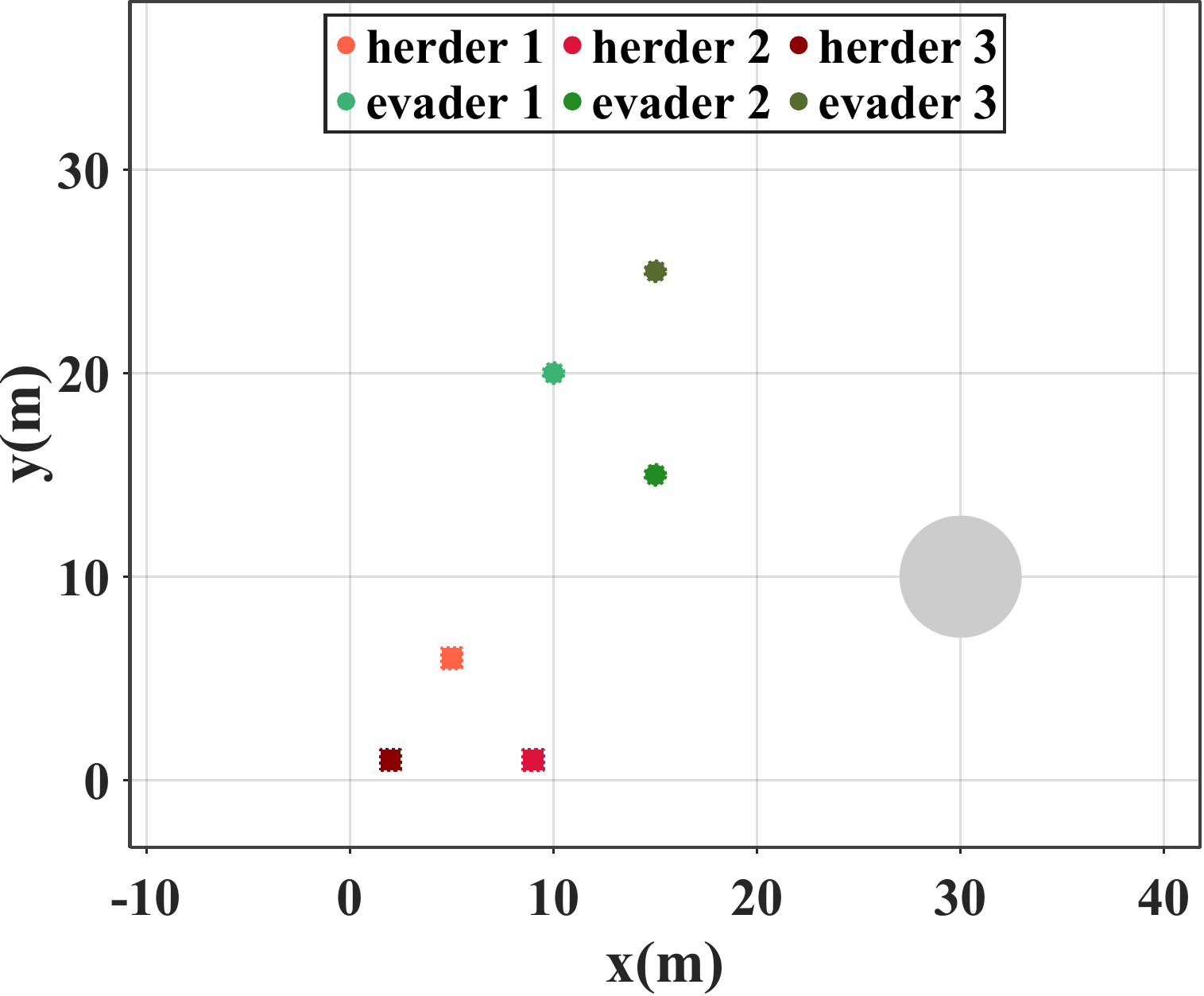}
        \caption{initial configuration}
        \label{herding_1}
    \end{subfigure}
    \hfill
    \begin{subfigure}[b]{0.32\textwidth}
        \includegraphics[width=\textwidth]{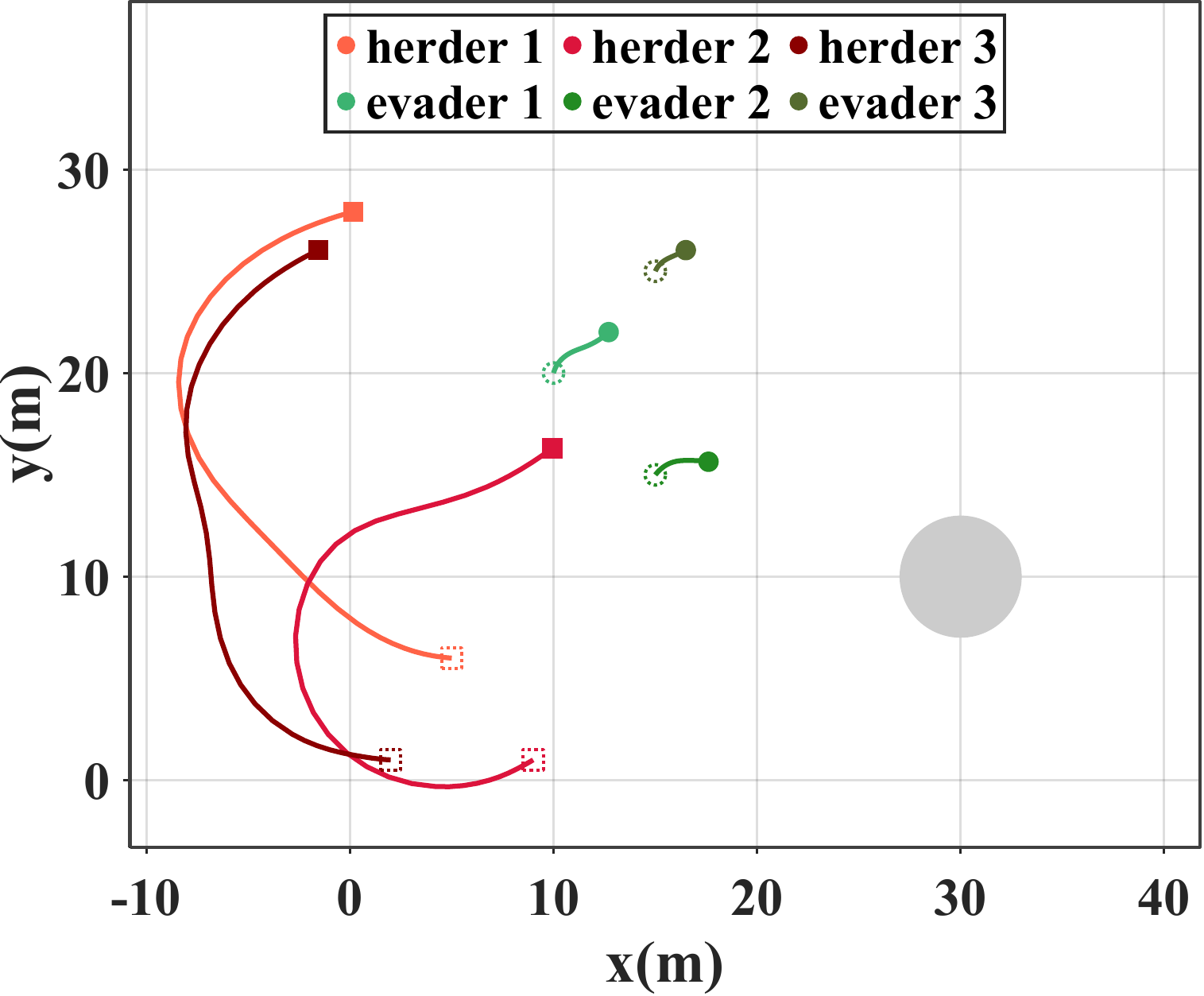}
        \caption{approaching goal region}
        \label{herding_2}
    \end{subfigure}
    \hfill
    \begin{subfigure}[b]{0.32\textwidth}
        \includegraphics[width=\textwidth]{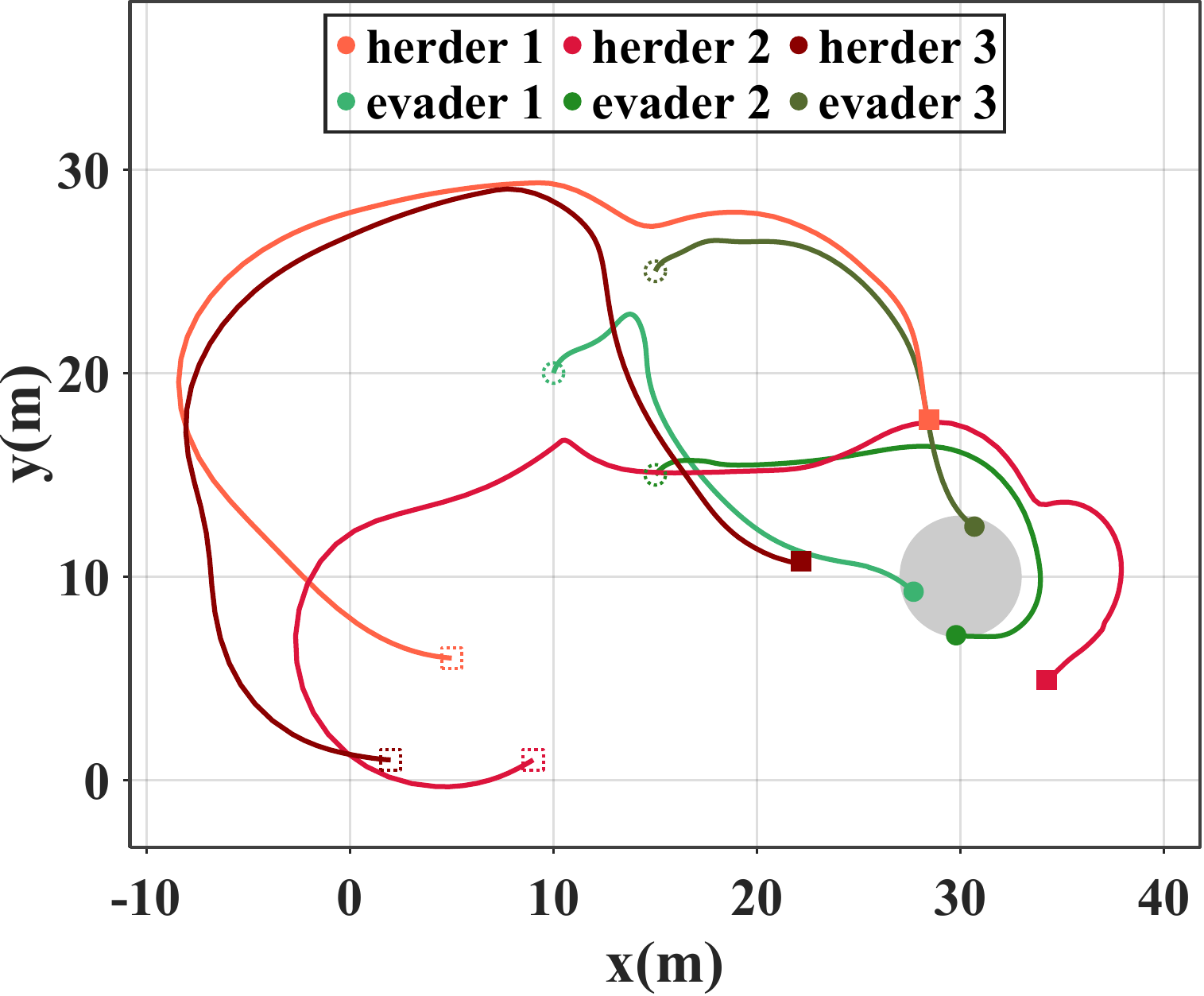}
        \caption{entered goal region}
        \label{herding_3}
    \end{subfigure}
    \hfill
    \begin{subfigure}[b]{0.32\textwidth}
        \includegraphics[width=\textwidth]{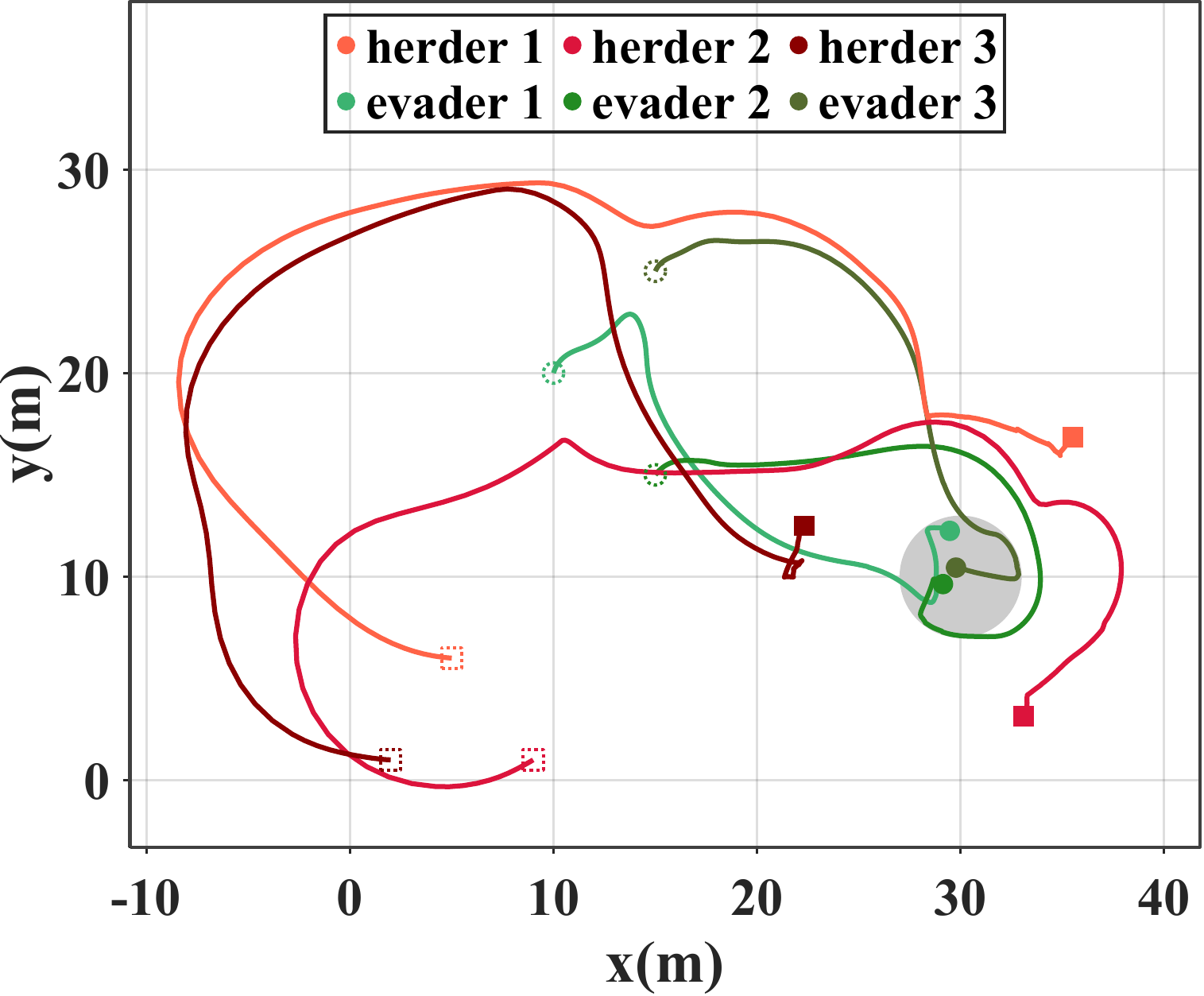}
        \caption{starting avoiding}
        \label{herding_4}
    \end{subfigure}
    \hfill
    \begin{subfigure}[b]{0.32\textwidth}
        \includegraphics[width=\textwidth]{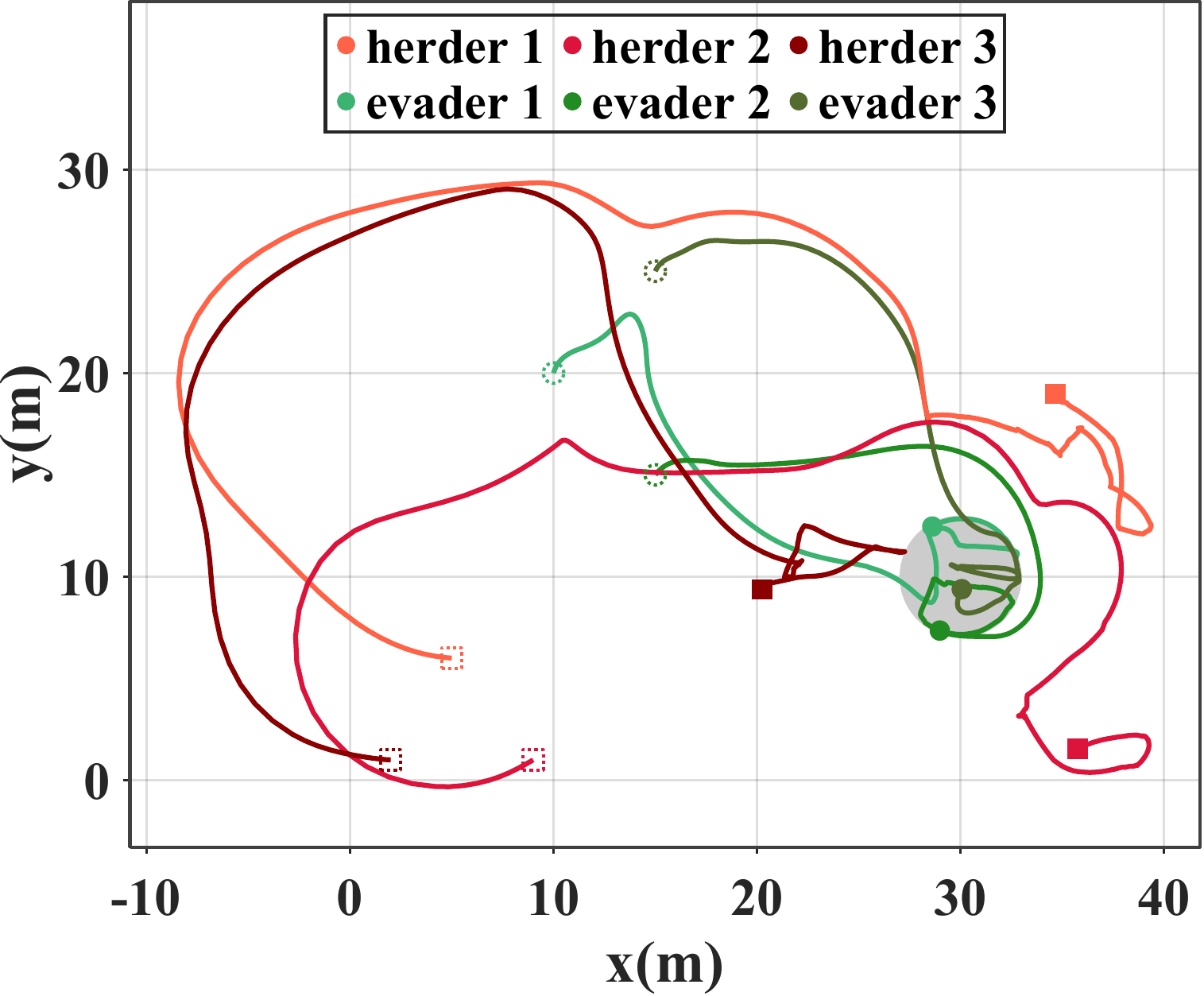}
        \caption{safe separation}
        \label{herding_5}
    \end{subfigure}
    \hfill
    \begin{subfigure}[b]{0.32\textwidth}
        \includegraphics[width=\textwidth]{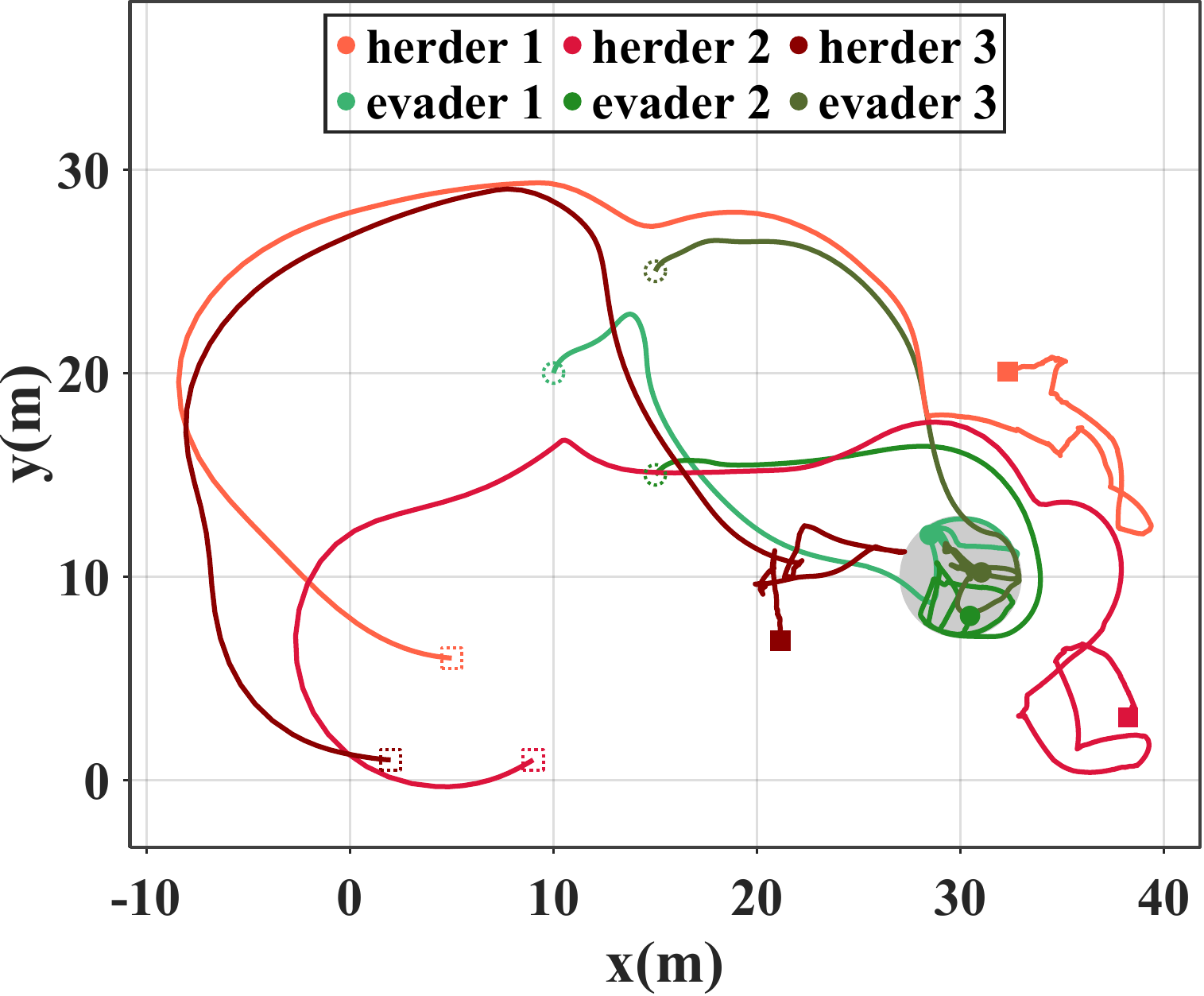}
        \caption{final configuration}
        \label{herding_6}
    \end{subfigure}
    \caption{Simulation results showing the trajectories of three herders (red squares) and three evaders (green circles) at different time instances.
    The initial position is represented with dashed squares and circles, and a gray circle represents the goal region.  }
    % Overall, the evaders are successfully herded into the goal region without collisions throughout the process. }
    \label{herding_sim}
\end{figure*}
% \vspace{-5pt} 

\vspace{3mm} 
\begin{Rmk}
    (Stability) The stability of the cooperative herding system with the controller $\mathbf{u}_{\rm{H}_k}$ in Eq.~\eqref{controller_kth_herder} is guaranteed when the CBF condition is inactive, i.e., when the nominal controller satisfies the CBF constraints in Eq.~\eqref{controller_kth_herder}. When stability cannot be ensured (typically due to active CBF constraints), an effective approach is to apply the safety filter in Eq.~\eqref{controller_kth_herder}, which helps preserve stability as much as possible.
\end{Rmk}

\vspace{3mm} 
\begin{Rmk}
    (Scalability) For Eq.~\eqref{controller_kth_herder}, the number of decision variables and the number of neighboring evaders, i.e., the size of $\mathcal{E}$, influence the complexity of the QP. Specifically, for Eq.~\eqref{controller_kth_herder}, the optimization problem has $2$ decision variables since $\mathbf{u}_{\mathrm{H}_{k}} \in \mathbb{R}^2$. The number of constraints depends on the number of neighboring evaders considered for the $i$th herder. In the worst case, this results in $m-1$ constraints. Therefore, the decentralized backstepping CBF is scalable to arbitrarily large groups of robots, with scalability dependent on $m$. Moreover, with the neighboring strategy, the value of $m$ can be reduced by appropriately selecting a constant neighboring distance, as discussed in~\cite{wang2017safety}.
\end{Rmk}

\vspace{3mm} 
\begin{Rmk}
    (Multiple equilibria) For Eq.~\eqref{controller_kth_herder}, the attractive force from the stabilizing controller may counteract the repulsive force from the CBF, potentially leading to the multiple-equilibria phenomenon, as discussed in~\cite{reis2020control,tan2024undesired}. While mitigating the multiple-equilibria issue is not the main focus of our paper, available strategies in~\cite{tan2024undesired} can be adopted to resolve the issue. In our implementations, we address this issue by introducing a small perturbation to the control input to effectively resolve it in practice.
\end{Rmk}

\vspace{3mm} 
Currently, the proposed algorithm focuses on the special case where the number of herders equals that of evaders. This modeling choice is not due to a fundamental limitation of the proposed framework, but rather a simplifying assumption introduced to emphasize the core technical contribution, namely, the development of a safety-critical herding control strategy based on backstepping CBFs. To facilitate controller design and analysis, we assume each herder only influences its nearest evader, thereby ensuring a one-to-one repulsion structure. 
This enables safe and effective herding without introducing the additional complexity of dynamic assignment or coordination strategies among herders. While this assumption is relatively strong, it allows us to isolate and highlight the core mechanisms of the proposed framework. Moreover, we acknowledge the importance of more general scenarios where $m\neq n$, and plan to relax this assumption in future work by incorporating dynamic coordination and multi-target assignment strategies. Ultimately, the current formulation serves as a meaningful and analytically tractable starting point for advancing effective and safe herding.

\begin{algorithm}[H]
\caption{Backstepping CBFs Herding}
\label{Algorithm1}
\begin{algorithmic}[1]
\Statex \hspace{-\algorithmicindent} \textbf{Initialization:} $\kappa_{\rm{H}}, \gamma_h, \gamma_a, k_h, k_a, \mu$
\Statex \hspace{-\algorithmicindent} for $k\in \mathcal{K}$, the $k$-th herder performs as follow:
% Not all evaders are in $\mathcal{G}$ 
% $V(\bf{x}_{\rm{E}_{i}})\neq 0, \forall i\in\mathcal{E}$
\While{$h_1(\bf{x}_{\rm{E}_{i}})\neq 0$}
    % \State \textbf{Herding}:
    \State  Update $\mathbf{r}_h(\mathbf{x}_{\rm{E}_{i}})$, $h_1(\mathbf{x}_{\rm{E}_{i}},  \mathbf{v}_{E_{ij}})$
    \State  Calculate $a, \mathbf{b}$
    \State Get nominal controller $\mathbf{u}_{\rm{H}_{k}}^{\text{nom}}\rightarrow \mathbf{u}_{{\rm{H}}_{k}}$
    \While{$h_2(\bf{x}_{\rm{E}_{ij}})<0$} % , \forall i, j \in\mathcal{E}, i\neq j 
        % \State \textbf{Collision Avoidance}: 
        \State Update $\mathbf{r}_a(\mathbf{x}_{\rm{E}_{ij}}), h_2(\mathbf{x}_{E_{ij}}, \mathbf{v}_{E_{ij}})$
        \State Calculate $ c , \mathbf{d} $
        \State \textbf{QP}: $\min_{\mathbf{u}_{\rm{H}_{k}}} \|\mathbf{u}_{\rm{H}_{k}}-\mathbf{u}_{\rm{H}_{k}}^{\text{nom}}\|^{2}$ $\xrightarrow{\text{}}$  $\mathbf{u}_{{\rm{H}}_{k}}$
    % \Else
        % \State  
    \EndWhile
    \State Update $\mathbf{u}_{{\rm{H}}_{k}}$
\EndWhile
\end{algorithmic}
\end{algorithm}

% -------------------------------------------------------------------------------------------------------------
\section{Simulation Results and Experimental Validation}\label{sec5}
To validate the effectiveness of the proposed algorithm in Section~\ref{sec4} in simultaneously goal reaching and avoiding collisions within evaders, we provide simulation results in MATLAB and further demonstrate its practical performance through real-robot experiments.

\subsection{Simulation Results}
This simulation investigates a herding scenario involving three herders and three evaders. 
The key parameters of the proposed algorithm are configured as follows: 
$\kappa_{\rm{H}}=10$, $\gamma_h=\gamma_a=0.5$, $k_h=k_a=1$, $\mu=1$, $R_{\text{avoid}}=0.5$. 
A velocity saturation with $v_{max} =3 $ m/s is enforced on each agent.
The goal region is defined as a circle centered at $\bf{x}_{\text{goal}}=(30, 10)$ with radius $R_{\text{goal}}=3.5$, which is sufficiently large to accommodate all evaders without congestion. 
Fig.~\ref{herding_sim} illustrates the trajectories of the herders and evaders at different time instances during the simulation.
 
\setcounter{figure}{3}
\begin{figure}[htp]  
% \captionsetup{width=0.45\textwidth}
    \centering
    \begin{subfigure}[b]{0.42\textwidth}
    % \captionsetup{width=1\textwidth}
        \includegraphics[width=\textwidth]{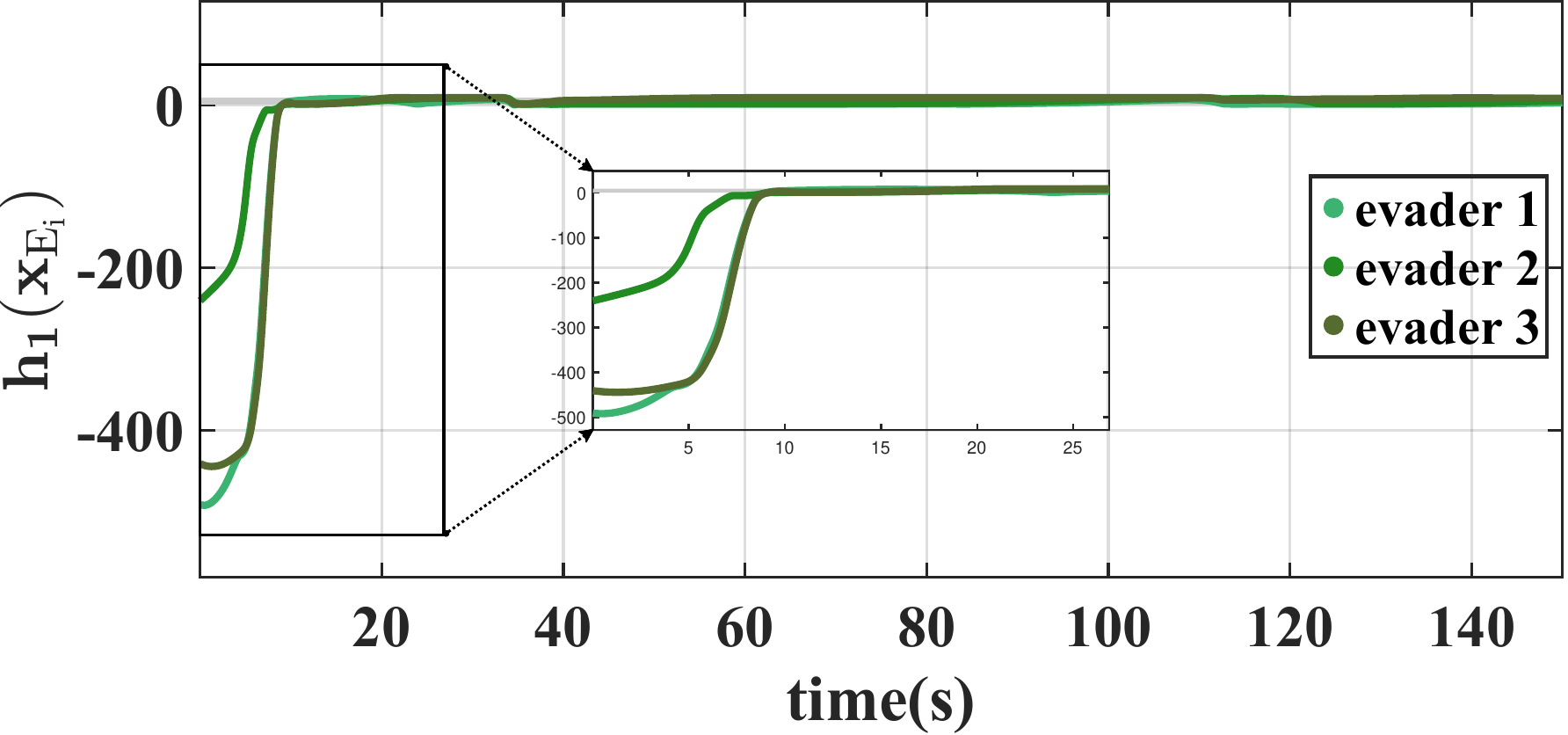}
        \caption{Evolution of the goal reaching CBF $h_1(\bf{x}_{\rm{E}_{i}}) $ values for all evaders.}
        \label{h1}
    \end{subfigure}
    \hfill
    \begin{subfigure}[b]{0.42\textwidth}
    % \captionsetup{width=1\textwidth}
        \includegraphics[width=\textwidth]{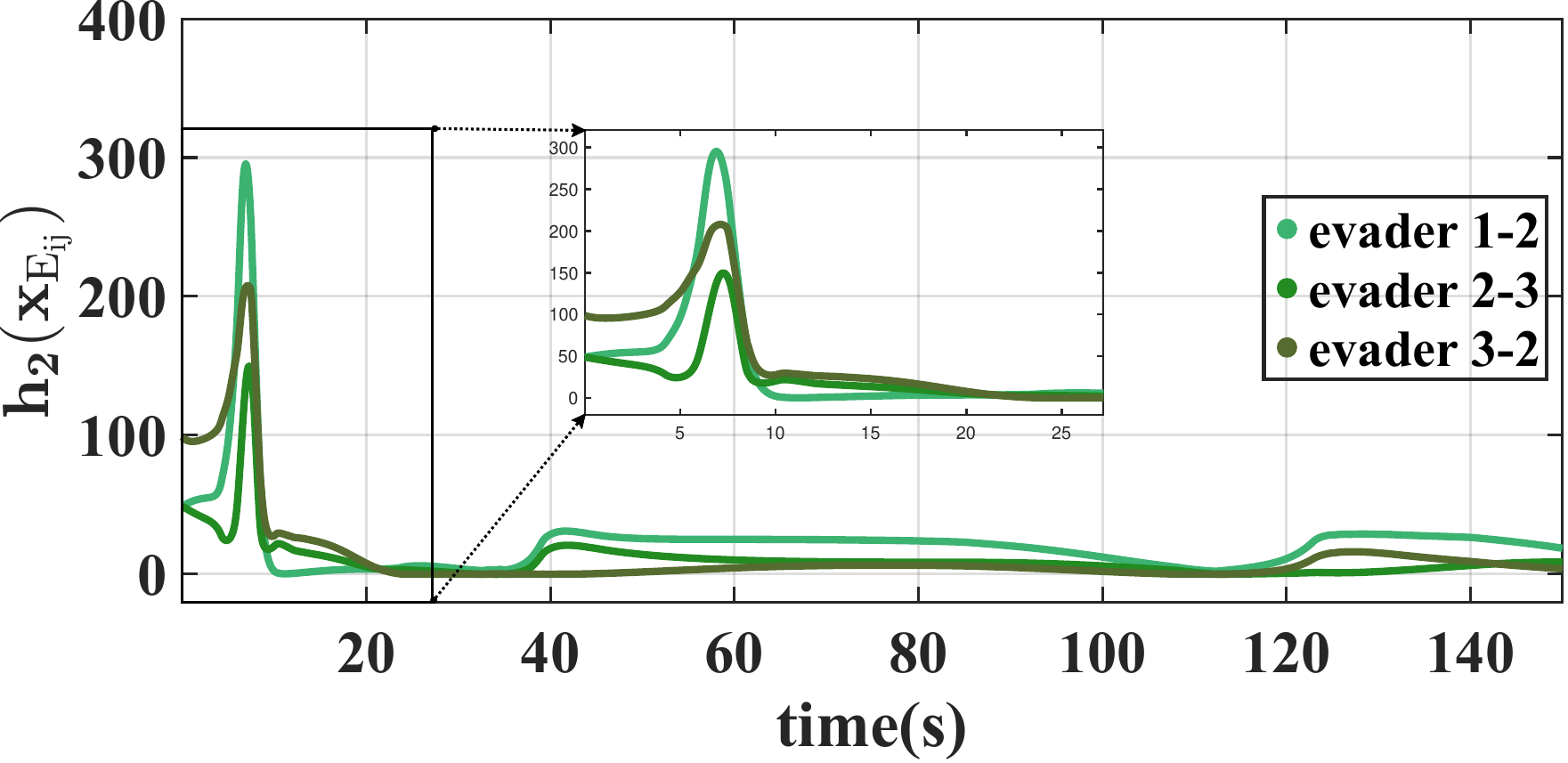}
        \caption{Evolution of the collision avoidance CBF $h_2(\bf{x}_{\rm{E}_{ij}})$ values between all pairs of neighboring evaders.}
        \label{h2}
    \end{subfigure}
     \caption{Performance metrics curves obtained from simulation results, evaluating the herding effectiveness and collision risk of evaders.}
    \label{h1_h2}
\end{figure} 
% \vspace{-5pt} 
In the initial phase of herding, the robots are positioned as shown in Fig.~\ref{herding_1}. 
Subsequently, three herders take a detour and move to the positions behind the three evaders in Fig.~\ref{herding_2}, which are opposite the goal region. This alignment ensures that the repulsive force exerted by the herders on the evaders is directed toward the goal region.
As the herding continues in Fig.~\ref{herding_3}, the evaders are gradually herded into the goal region, with both the evaders' and herders' speeds decreasing accordingly. 
Since the repulsive forces of the herders persist, the evaders continue to move within the goal region. 
Notably, in Fig.~\ref{herding_4}, although the evaders are in extremely close proximity, collisions are consistently avoided throughout the process. 
The herders cooperatively maneuver around the evaders, using only the combined repulsive forces to ensure that the evaders do not collide and remain safely separated in Fig.~\ref{herding_5}. 
In the final stage in Fig.~\ref{herding_6}, when evaders risk colliding or exiting the region, the herders execute agile maneuvers, such as sharp turns and repositioning, to proactively mitigate these situations. 
Throughout the process, the herders exhibit coordinated behavior with minimal control effort, effectively maintaining the evaders within the goal region while ensuring their safety. These results validate the practicality of the proposed strategy under simulation.

Fig.~\ref{h1_h2} shows the evolution curves of two performance metrics $h_1(\bf{x}_{\rm{E}_{i}}) $ and $h_2(\bf{x}_{\rm{E}_{ij}})$, derived from the above constructed CBFs. 
These metrics are designed to quantify the success of goal reaching and the safety of collision avoidance, respectively.
As depicted in Fig.~\ref{h1}, the value of $h_1(\bf{x}_{\rm{E}_{i}}) $ increases as the distance between the evaders and the goal region decreases. 
Once all evaders successfully reach the goal region, the values stabilize within the interval $[0, R_{\text{goal}}^2]$, indicating that the goal-reaching condition has been fulfilled.
In parallel, Fig.~\ref{h2} shows the evolution of $h_2(\bf{x}_{\rm{E}_{ij}})$, which initially takes higher values due to the significant separation between the evaders. 
As the evaders are herded toward the goal region, they may tend to become more dispersed, resulting in an increase in $h_2(\mathbf{x}_{\rm{E}{ij}})$. 
Notably, when the evaders are successfully herded into and remain in close proximity to one another, the proposed algorithm effectively guarantees that $h_2(\bf{x}_{\rm{E}_{ij}})$ remains strictly positive. 
From the simulation results, all evaders are safely herded into the target region through the repulsive forces generated by the herders' controlled motion, confirming the effectiveness and safety of the algorithm in herding.

\subsection{Experimental Validation}
The experiments were conducted in a 5×5 meter indoor environment equipped with a Vicon motion capture system, which involved three herder robots and three evader robots.
The goal region was defined as a circular area centered at the origin, with a radius of 0.6 meters.
As illustrated in Fig.~\ref{fig_car}, all robots are omnidirectional platforms equipped with Mecanum wheels, enabling holonomic motion and direct velocity control. 
Their motion characteristics are well approximated by first-order integrator dynamics, consistent with the model used in the simulations.
Each robot has a physical diameter of 0.3 meters and a maximum velocity of 0.3 m/s.
Further details on the hardware design and control interfaces can be found in~\cite{ma2024omnibot}. 
\setcounter{figure}{4}
\begin{figure}[http] % htp
    \centering
    % \makebox[0pt]{%
    \includegraphics[width=0.4\textwidth]{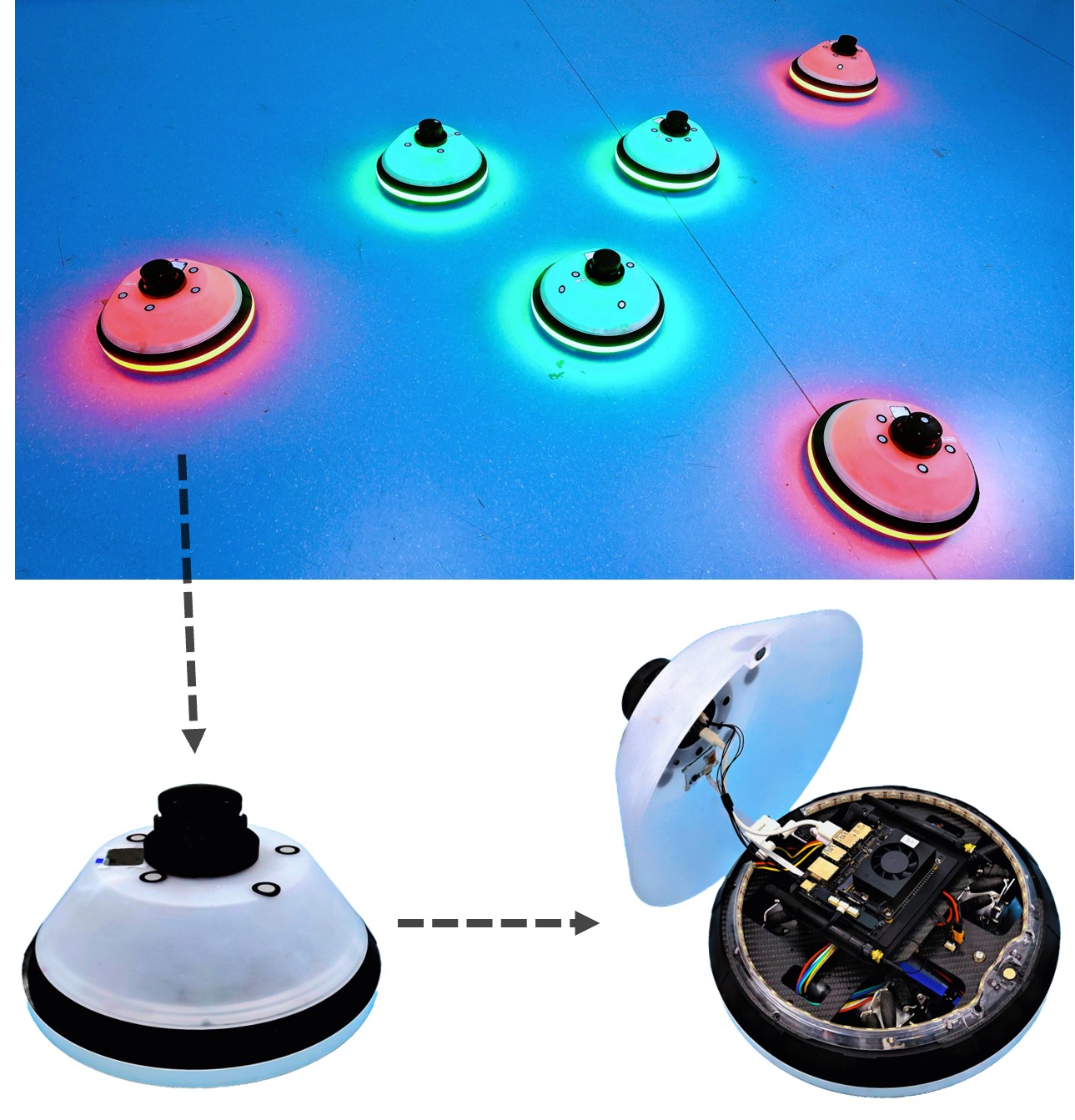}
    % \vspace{10pt} 
    \caption{The omnidirectional robot used in the experiments.}
    \label{fig_car}
\end{figure}

\setcounter{figure}{6}
\begin{figure*}[!http] % 星号表示跨两栏，通常在页顶显示
    \centering
    \begin{subfigure}[b]{0.32\textwidth}
        \includegraphics[width=\textwidth]{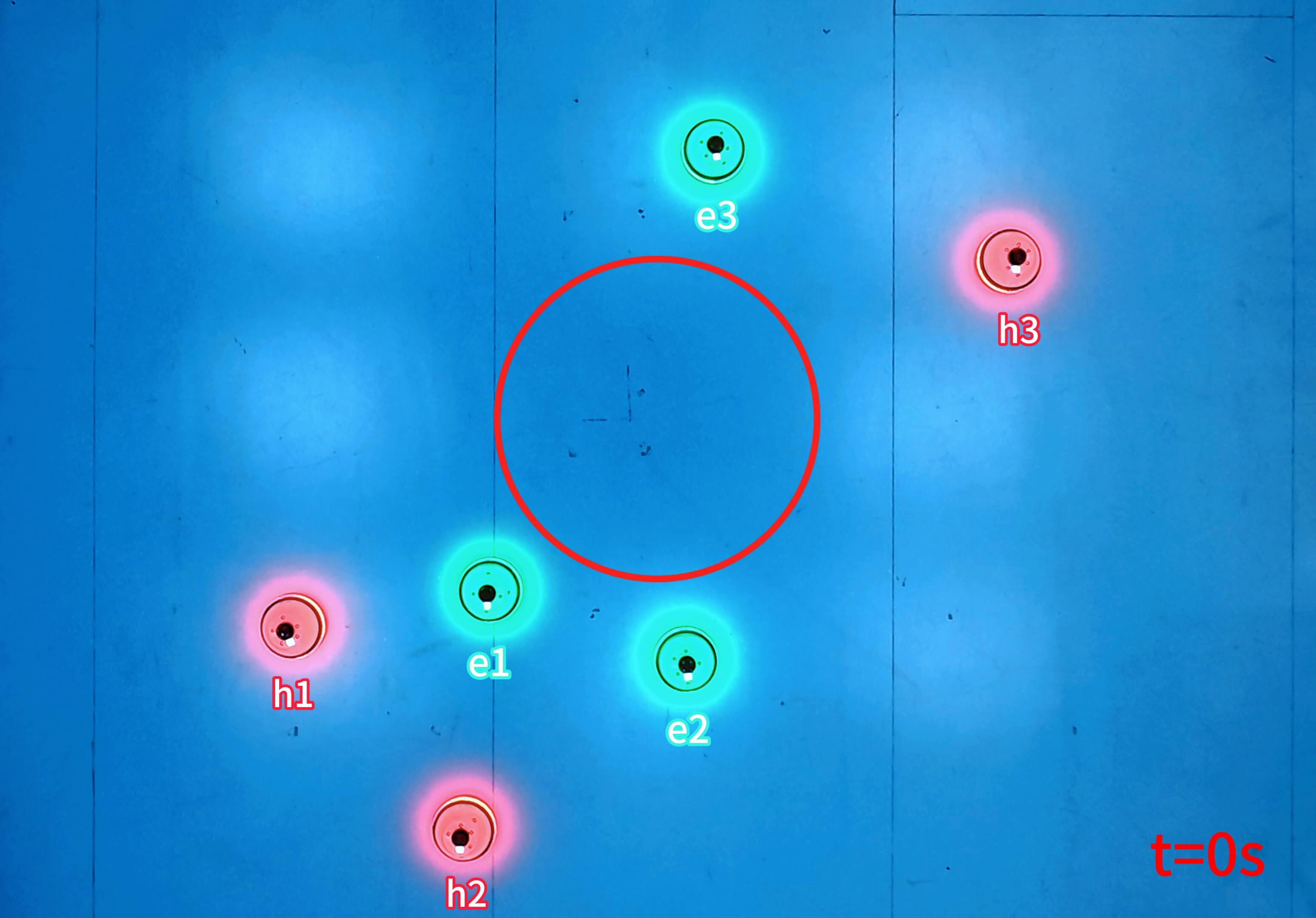}
        \caption{initial configuration}
        \label{real_1}
    \end{subfigure}
    \hfill
    \begin{subfigure}[b]{0.32\textwidth}
        \includegraphics[width=\textwidth]{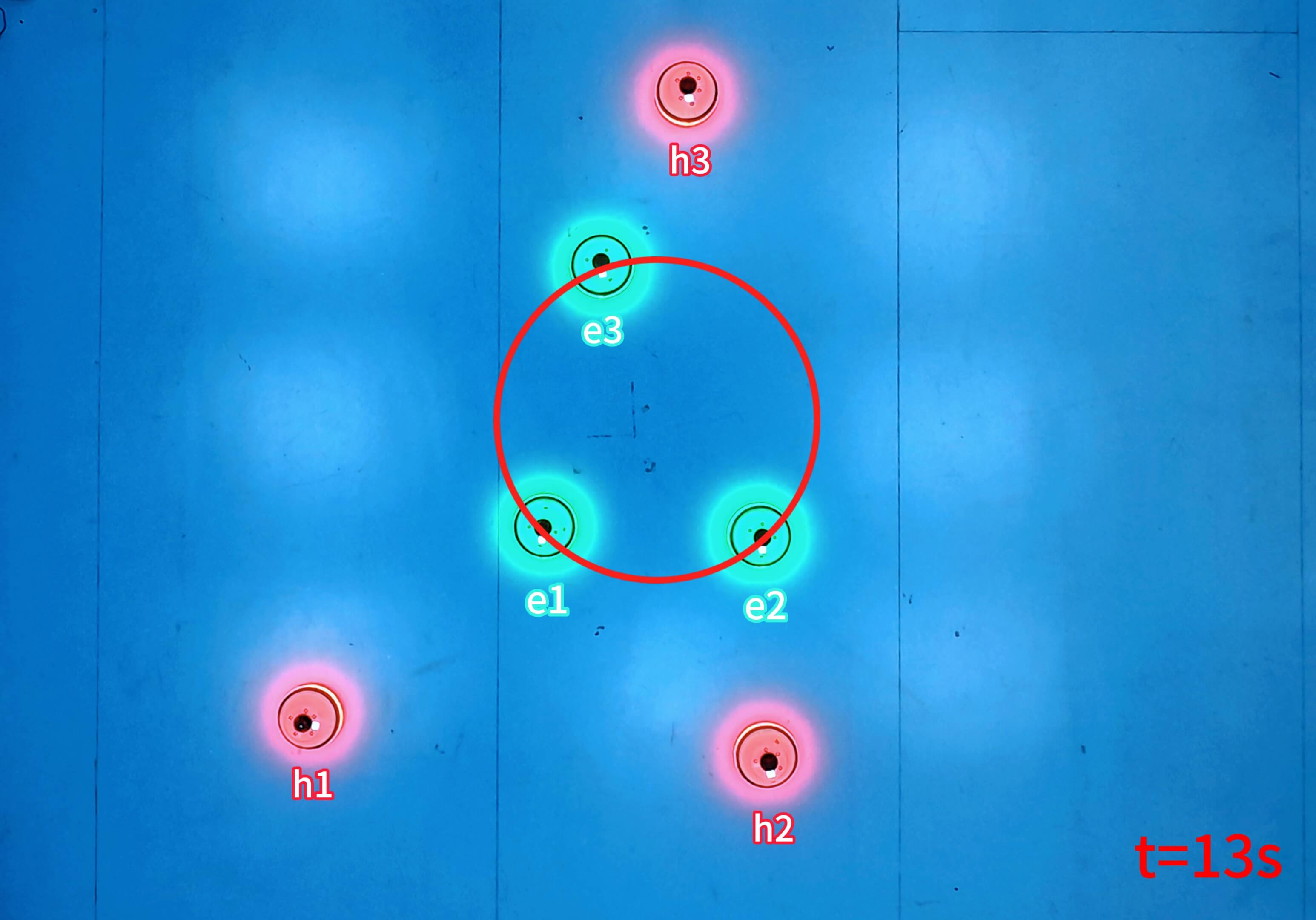}
        \caption{approaching goal region}
        \label{real_2}
    \end{subfigure}
    \hfill
    \begin{subfigure}[b]{0.32\textwidth}
        \includegraphics[width=\textwidth]{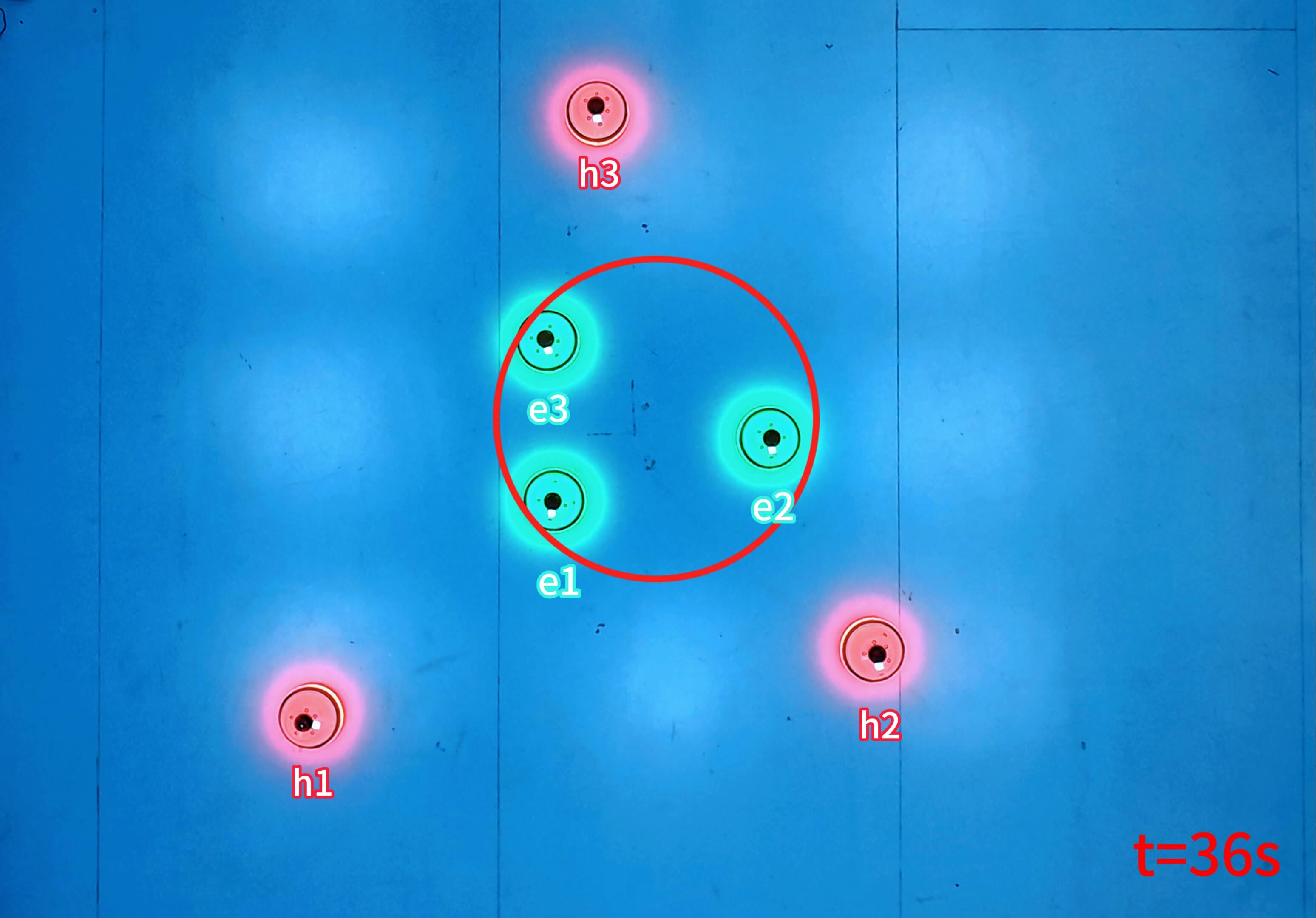}
        \caption{entered goal region}
        \label{real_3}
    \end{subfigure}
    \hfill
    \begin{subfigure}[b]{0.32\textwidth}
        \includegraphics[width=\textwidth]{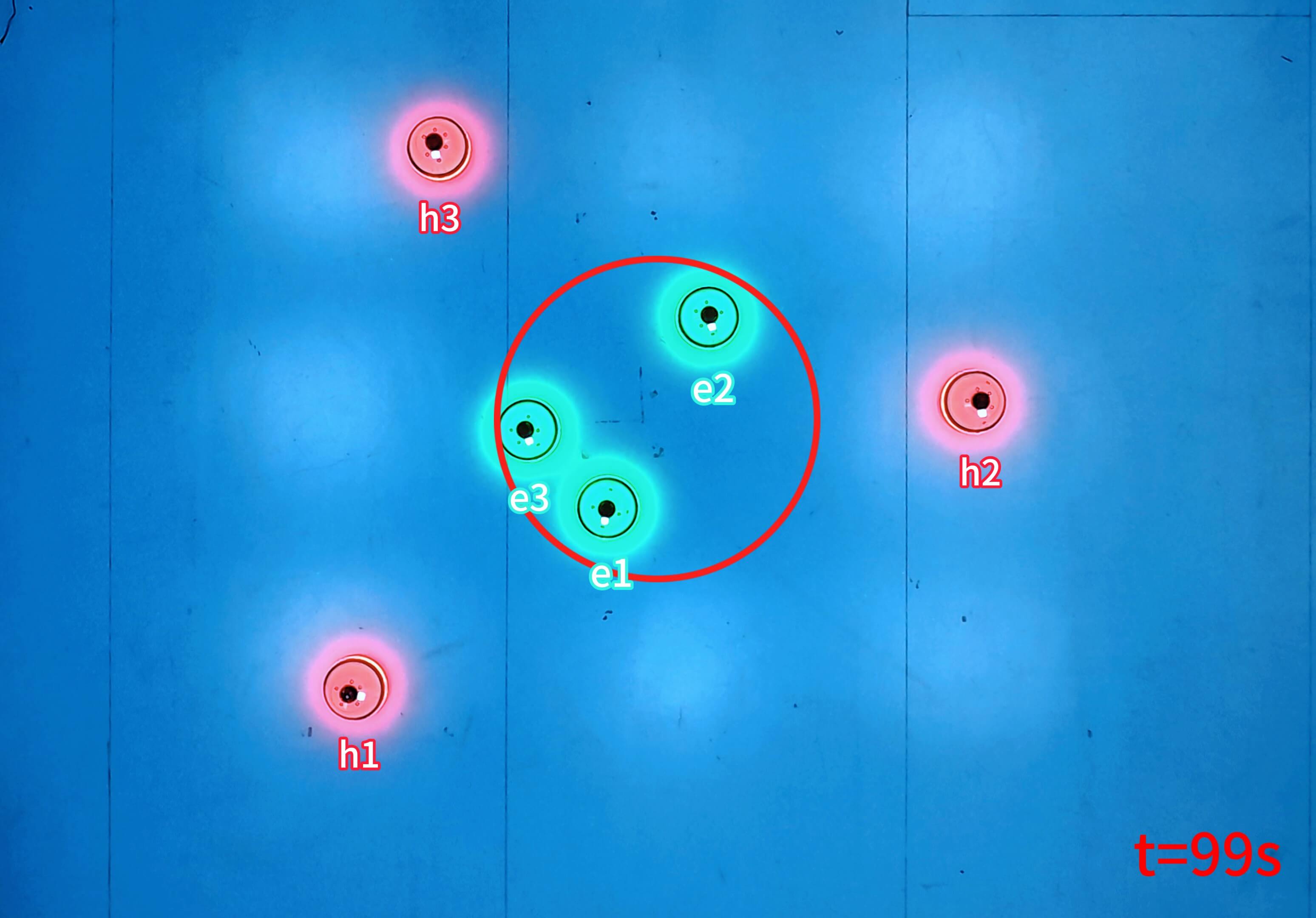}
        \caption{starting avoiding}
        \label{real_4}
    \end{subfigure}
    \hfill
    \begin{subfigure}[b]{0.32\textwidth}
        \includegraphics[width=\textwidth]{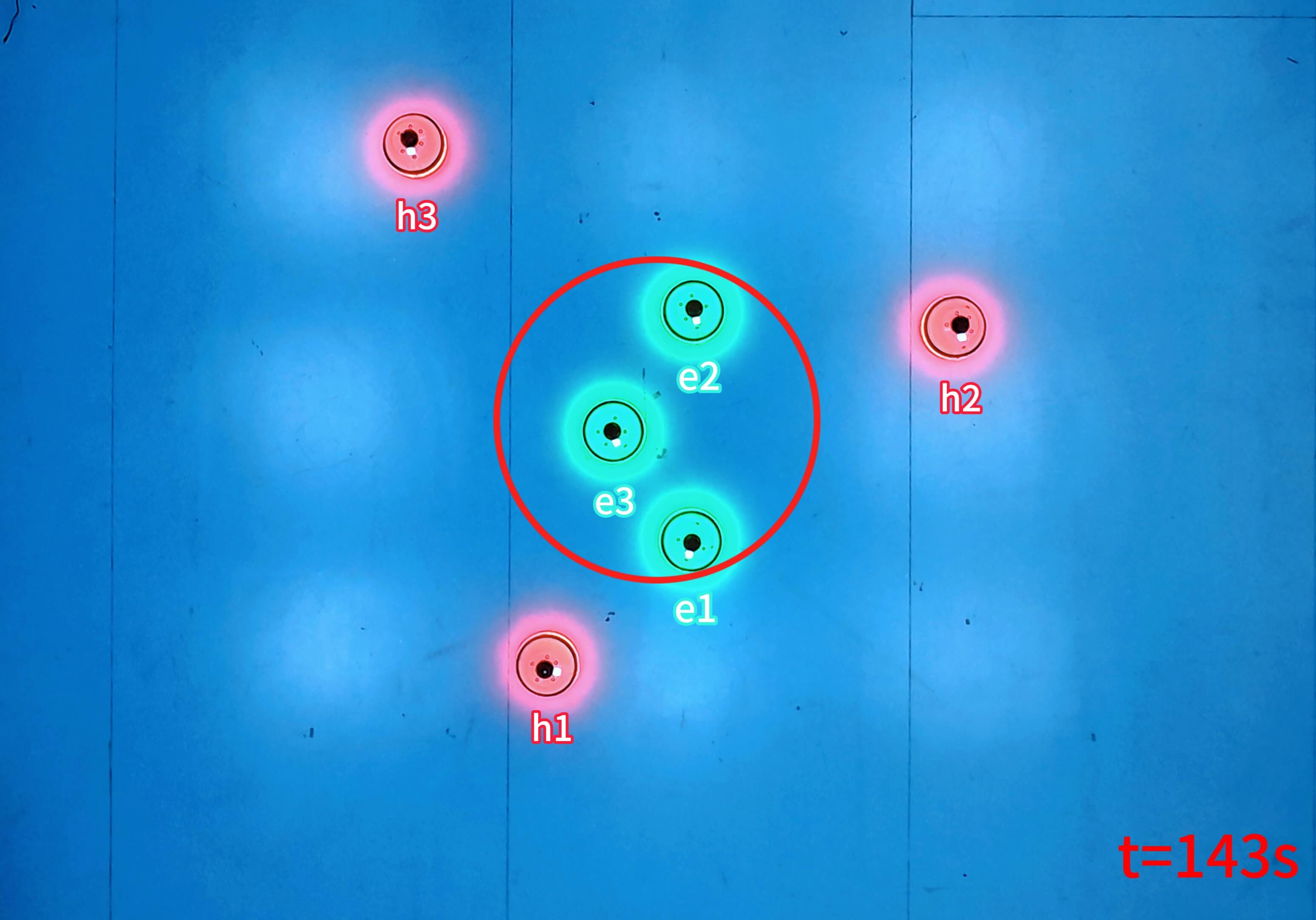}
        \caption{safe separation}
        \label{real_5}
    \end{subfigure}
    \hfill
    \begin{subfigure}[b]{0.32\textwidth}
        \includegraphics[width=\textwidth]{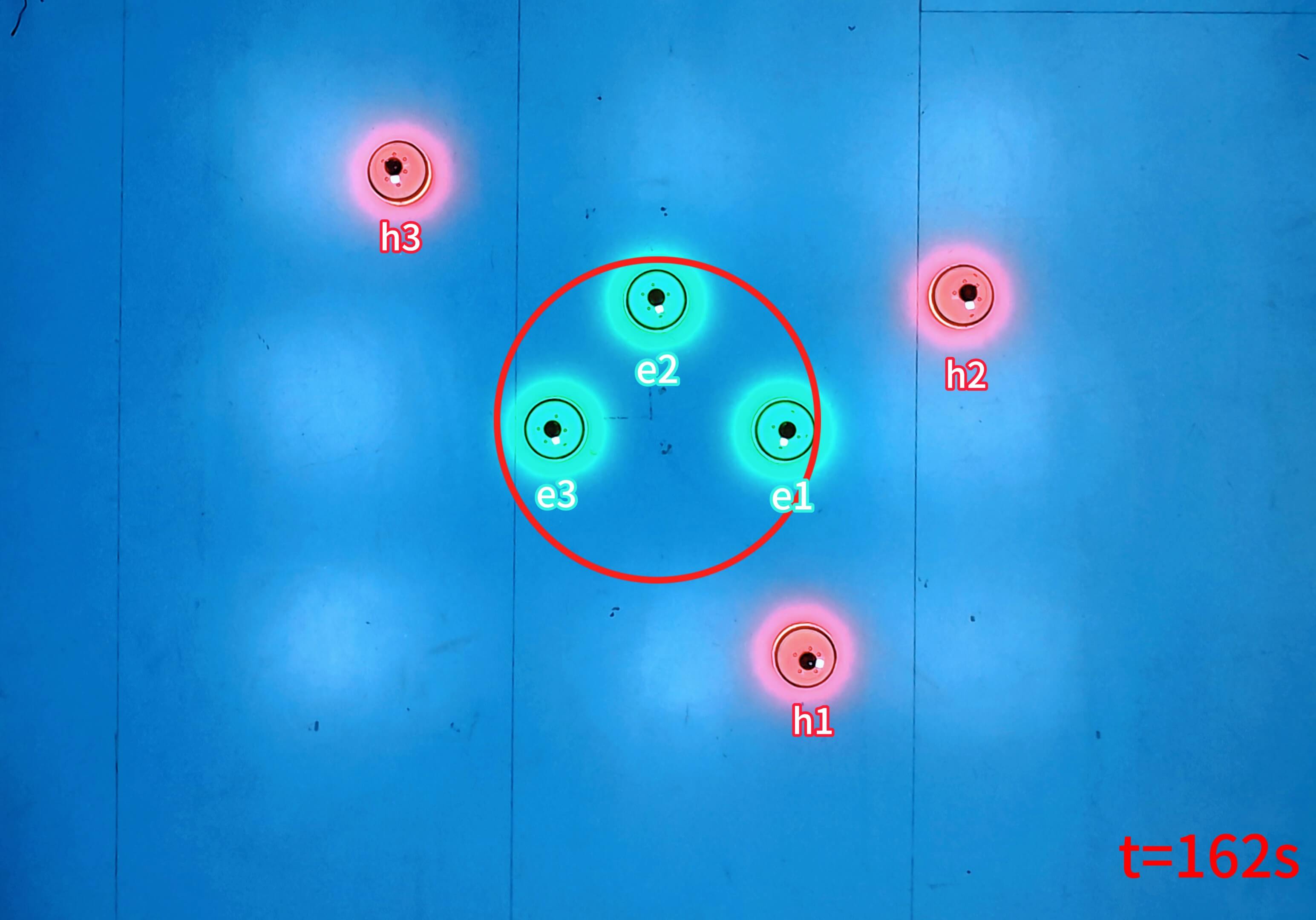}
        \caption{final configuration}
        \label{real_6}
    \end{subfigure}
    \caption{Experiment results of three herder robots herding three evader robots. The red and green robots respectively represent the herder and evader robots, and the red circular area denotes the goal region. 
    The experimental video is available at: \url{https://youtu.be/wGbxvfTbPL0}.}
    \label{herding_exp}
\end{figure*}
In our experimental setup, the user control station continuously monitors the positions of all robots in real time using a Vicon motion capture system.
The control station executes the proposed algorithm at each control cycle, which computes the required control inputs for each robot based on the most up-to-date positional data. 
These control commands are then transmitted to the robots via the Message Queuing Telemetry Transport(MQTT) protocol, with each robot executing the commands through its onboard client system.
The control cycle operates at a fixed frequency of 20 Hz. At the end of each cycle, the control station receives updated position measurements, processes the latest feedback, and updates the control inputs accordingly, seamlessly transitioning into the next cycle. 
This closed-loop system ensures continuous monitoring, real-time synchronization, and coordinated adjustment of all robot behaviors.

\setcounter{figure}{5}
\begin{figure}[http]  
% \captionsetup{width=0.42\textwidth}
    \centering
    \begin{subfigure}[b]{0.408\textwidth}
    % \captionsetup{width=1\textwidth}
        \includegraphics[width=\textwidth]{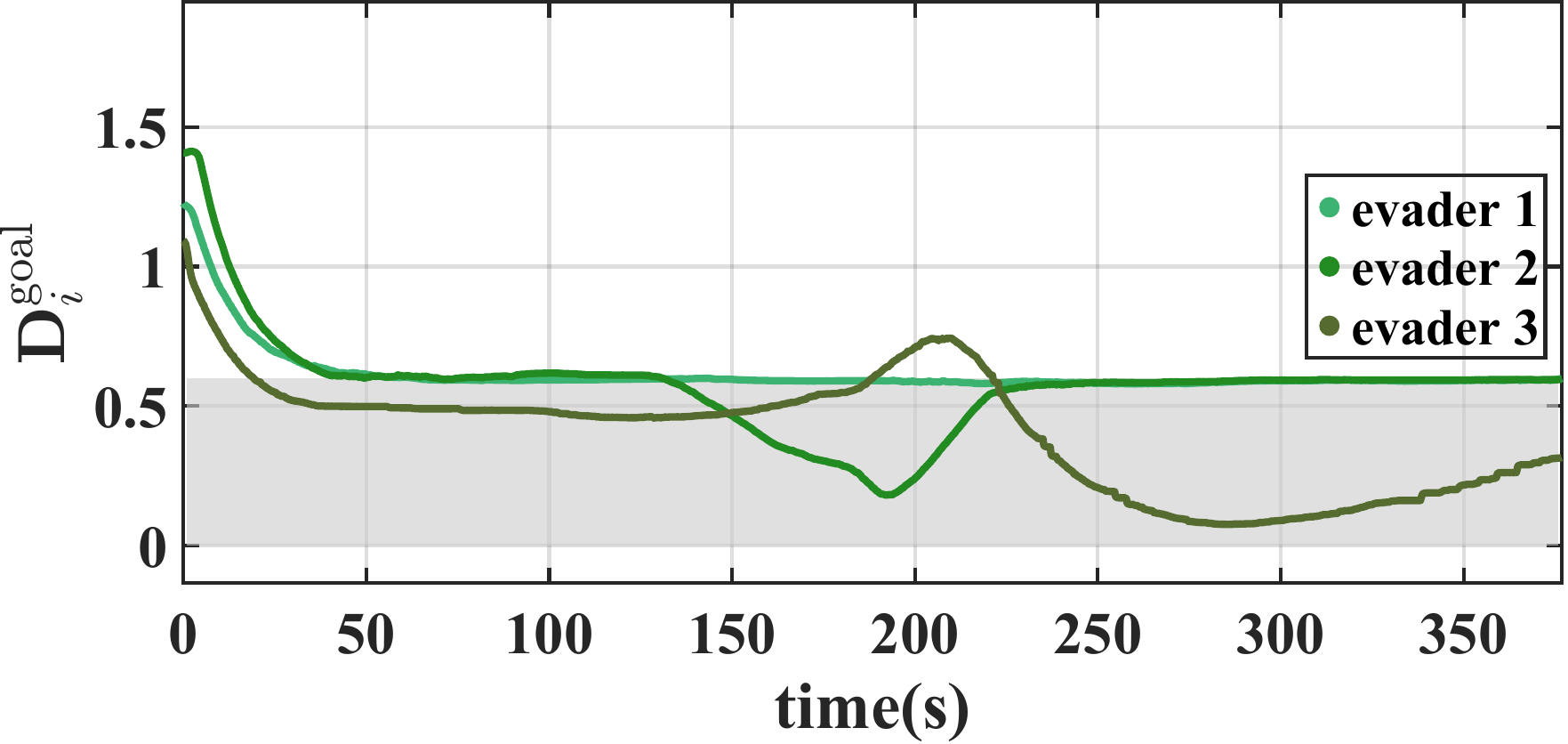}
        \caption{Distance of each evader to the goal region over time.}
        \label{D_goal}
    \end{subfigure}
    \hfill
    \begin{subfigure}[b]{0.408\textwidth}
    % \captionsetup{width=1\textwidth}
        \includegraphics[width=\textwidth]{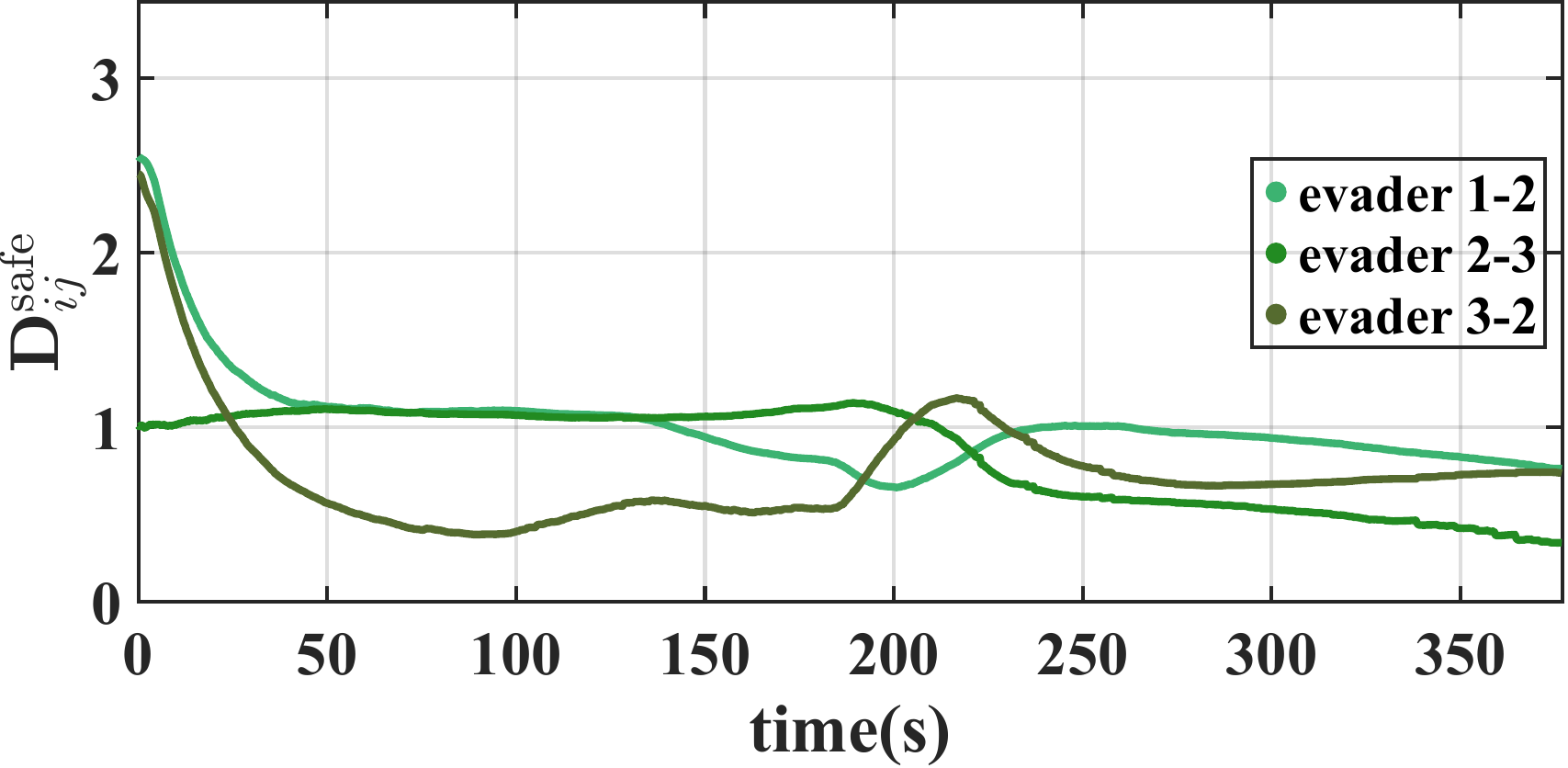}
        \caption{Pairwise distances between neighboring evaders over time.}
        \label{D_safe}
    \end{subfigure}
     \caption{Evolution of key performance metrics obtained from real-world experiments, characterizing goal reaching effectiveness and evader collision avoidance.}
    \label{D_goal_D_safe}
\end{figure}
% \vspace{-10pt} 

Fig.~\ref{herding_exp} illustrates the performance of the proposed algorithm in a real-world robot experiment. 
Due to the limited physical space, the initial configurations of the herders and evaders were relatively compact compared to those in the simulation.
This spatial proximity reduced the likelihood of large reorientation maneuvers during the initial alignment phase, thereby streamlining the overall herding process.
Despite the more constrained initial setup, the experimental results clearly demonstrate that the proposed strategy successfully guides multiple evaders into the designated goal region while ensuring safe separation among them throughout the herding process. 

To quantitatively assess the real-world performance of the proposed herding strategy, we examine two key metrics: the real distance of each evader to the goal region and the pairwise distances among evaders. 
Specifically, Fig.~\ref{D_goal} presents the distance $D^{goal}_{i}$ between each evader and the goal region over time. 
As the trajectories evolve, all three evaders are successfully herded into the goal region, with their distances eventually stabilizing within the predefined experimental radius of 0.6 meters, thereby confirming successful convergence. 
Similarly, Fig.~\ref{D_safe} presents the pairwise distances $D^{safe}_{ij}$ between neighboring evaders. 
These distances consistently remain strictly positive throughout the experiment, demonstrating that the proposed strategy reliably maintains safe separation among evaders, even as they converge toward the goal region, thereby ensuring effective collision avoidance under spatial constraints. 
Overall, these results validate the effectiveness and safety of the proposed approach in simultaneously achieving the dual objectives of goal reaching and collision avoidance within evaders under the real-world conditions. 

\section{Conclusion}\label{sec6}
In this paper, we propose a multi-robot cooperative herding strategy through backstepping CBFs. 
To address the challenges arising from the underactuated nature of the herding system involving non-cooperative agents, we first reformulate the herding system dynamics into a control-affine structure compatible with the backstepping method. Then, we construct the backstepping CBFs separately for the dual objectives of goal reaching and collision avoidance. 
Finally, we provide both centralized and decentralized cooperative controllers to implement the proposed herding strategy. 

While the proposed method demonstrates promising performance, it still exhibits certain limitations.
For instance, it assumes full observability of evader positions and accurate modeling of their dynamics, which may not hold in real-world scenarios. Model uncertainties and partial observability could degrade system performance, underscoring the importance of incorporating robust control strategies, adaptive mechanisms, and state estimation techniques.
Moreover, real-world deployment necessitates onboard sensing and autonomous localization capabilities. It is essential to integrate perception modules, such as cameras or LiDAR, into both herder and evader robots to enable real-time, decentralized operation without relying on external infrastructure.
These limitations point to several promising directions for future research, alongside extensions to more complex scenarios, including situations where the number of herders exceeds that of evaders or where the herders have lower mobility than the evaders.

\backmatter

\bmhead{Author contribution}
Kang Li and Ming Li made substantial contributions to the conception and design of the work and wrote the main manuscript text. Kang Li conducted the primary simulations and experiments and provided the curve analysis figures. Wenkang Ji assisted Kang Li with the real-robot experiments and organized the references. Shiyu Zhao and Zhiyong Sun critically reviewed the manuscript for important intellectual content and assume collective responsibility for ensuring the accuracy and integrity of the work.

\bmhead{Funding}
This work was supported by the National Natural Science Foundation of China (Grant No. 62473320).

\bmhead{Data availability}
No datasets were generated or analyzed during the current study.

\bmhead{Declarations}

\bmhead{Conflict of interest}
The authors declare that they have no conflict of interest.

\bibliography{herding_ref} 
\end{document}